\pdfoutput=1

\pdfoutput=1

\documentclass[11pt]{scrartcl}
\usepackage[a4paper, total={16cm, 24cm}]{geometry}
\usepackage{microtype}
\usepackage{authblk}
\title{Pareto Optimality in Approval-Based Multiwinner Voting}
\date{\vspace{-1.5cm}}

\author[1]{Joshua Schünke}
\affil[1]{Hasso Plattner Institute, University of Potsdam, Germany}
\affil[ ]{\texttt{joshua.schuenke@student.hpi.de}}

\usepackage{xcolor}
\usepackage{graphicx}
\usepackage{pgfplots}
\pgfplotsset{compat=1.17}

\usepackage{amsmath, amssymb, amsthm, mathtools, nicefrac, siunitx}

\usepackage{thmtools}
\usepackage{thm-restate}
\newtheorem{theorem}{Theorem}[section]

\newtheorem{proposition}[theorem]{Proposition}
\newtheorem{lemma}[theorem]{Lemma}
\newtheorem{corollary}[theorem]{Corollary}
\theoremstyle{definition}
\newtheorem{definition}[theorem]{Definition}

\newtheorem{conjecture}[theorem]{Conjecture}

\usepackage
[
    ruled,
    vlined,
    linesnumbered,
]
{algorithm2e}

\usepackage[pagebackref]{hyperref}
\hypersetup{
	pdfencoding=auto,
	psdextra,
	colorlinks=true,
	citecolor=green!50!black,
	linkcolor=red!60!black,
}

\usepackage[nameinlink]{cleveref}

\crefname{AlgoLine}{line}{lines}
\Crefname{AlgoLine}{Line}{Lines}

\usepackage{natbib}
\usepackage[inline]{enumitem}
\usepackage{booktabs}
\usepackage{subcaption}

\usepackage{tikz}
\usetikzlibrary{patterns}
\usetikzlibrary{patterns.meta}
\usetikzlibrary{backgrounds}

\newcommand{\restatehere}[1]{%
	\marginline{\vspace{0.6cm}\footnotesize \hyperlink{original#1}{\hypertarget{restated#1}{[Main]}}}%
	\csname #1\endcsname*%
}

\newcommand{\N}{\mathbb{N}}

\newcommand{\getcolor}[1]{\ifcase#1 orange!20\or red!40\or purple!35\or magenta!20\or violet!30\or blue!20\or cyan!20\or teal!20\or green!20\or lime!20\or yellow!40\else black\fi}

\newcommand{\drawcand}[3][]{%
    \pgfmathtruncatemacro{\row}{#2-1}%
    \def\chosen{#1}%
    \ifx\chosen\empty
        \pgfmathtruncatemacro{\idx}{#2-1}%
        \edef\chosen{\getcolor{\idx}}%
    \fi
  \foreach \s/\e in {#3} {%
    \draw[fill=\chosen] (\s-1,\row) rectangle (\e,\row+1);
    \node at (\s-0.5,\row+0.5) {$c_{#2}$};
  }%
}

\begin{document}

\maketitle

\bigskip
{\footnotesize\tableofcontents}

\newpage
\begin{abstract}
	\begin{center}
		\textbf{\textsf{Abstract}} \smallskip
	\end{center}
	In approval-based multiwinner voting, Pareto optimality is used as an axiom capturing efficiency of committees.
	We study the structure of the space of Pareto optimal committees in restricted domains and in general by investigating the monotonicity and reconfigurability of such committees.
	For the Candidate Interval and Voter Interval domains, we propose the \emph{Single Dominance Only} property, which provides a simple characterization of Pareto optimality, and show that Pareto optimal committees satisfy committee monotonicity using this property.
	Further, we show that, for the above domains, any Pareto optimal committee can be reconfigured into any other Pareto optimal target committee without using auxiliary candidates, meaning that the candidates in the starting but not the target committee can be replaced by candidates in the target but not the starting committee one by one while preserving Pareto optimality at every step.
	In addition, we adapt a polynomial-time algorithm for finding a committee satisfying EJR+, a proportionality axiom, such that it also satisfies Pareto optimality, for the above domains.
	We further describe a polynomial-time algorithm for counting the number of Pareto optimal committees for voting instances satisfying Voter Interval, and give a proof idea for its correctness.
	For the unrestricted domain, we explain the challenges of proving committee monotonicity and reconfigurability.
	We provide an example in which the distance of two committees in the Pareto optimality reconfiguration graph exceeds the distance proven for the above domains, and outline an approach toward showing the connectedness of the graph.

\end{abstract}

\section{Introduction}
Social choice theory includes the study of reaching collective decisions based on individual preferences through voting.
In an election, multiple entities, called the voters, can express their opinion on multiple options available, called the candidates.
One or more candidates are then elected as winners, based on the voters' opinions.
In the theory of social choice, approval-based voting is a method in which each voter either approves or disapproves each candidate \citep{Brams_Fishburn_1978}.
In approval-based multiwinner voting, the committee, a fixed number of candidates winning the election, is determined based on the approvals of all voters.

Approval-based multiwinner voting can be used to reach a collective selection from several options that cannot be implemented simultaneously.
For example, suppose a group of gardeners wants to plant exactly three trees and has many saplings they could potentially plant.
As every gardener has a different opinion on which three saplings to plant, they want to find a selection of saplings representing the gardeners opinions fairly.
They can do so by using approval-based multiwinner voting.
Each gardener can approve the saplings they want to plant and disapprove the others.
Based on these approvals, the three saplings to be planted, the committee, are selected.
However, in this example as well as in general, it is not obvious how a committee should be chosen to best represent all voters fairly.
In order to asses the fairness and efficiency of the committee elected, several axioms a committee can fulfill exist.
These axioms also imply a structure or well-behavedness of the committees.
For instance, committee monotonicity can hold for committees satisfying a certain axiom.
Committee monotonicity is said to hold for committees satisfying an axiom, if a candidate can be added to any committee such that the new committee also satisfies the axiom.

We will now introduce two axioms, one expressing a notion of efficiency and the other one of fairness.
Here, efficiency refers to how satisfied voters are with a committee.
The efficiency of a committee can, for example, be measured by summing the number of candidates each voter approves within that committee.
However, this measure is quite restrictive and oversimplified, as it does not account for differences between individual voters.

Pareto optimality is a notion of efficiency used in different fields of study such as multiobjective optimization \citep{po-multiobjective-optimization}, which can also be applied to approval-based multiwinner voting.
Generally speaking, an object is Pareto optimal if no other object is at least as good in every regard and better in some regards.
A committee $A$ is said to be Pareto optimal, if no different committee $B$ exists, such that every voter approves at least as many candidates in $B$ as in $A$, and at least one voter approves strictly more candidates in $B$ than in $A$.
Else, if such a committee $B$ exists, we say that committee $B$ dominates committee $A$.
Thus, if a committee is Pareto optimal, some voter will be less represented or every voter will be equally represented in any other committee.
As Pareto optimality thereby constitutes an intuitively sensible and simple to understand property, it is used as an efficiency axiom for committees as well.

Justified representation (JR) and variations of JR such as EJR+ on the other hand are recently studied proportionality axioms for committees in approval-based multiwinner voting \citep{JustifiedRepresentation,PropAxiomsForMultiwVoting}.
The JR axiom and its extensions are intended to ensure that sufficiently large groups of voters with similar interests, even if they are in the minority, are represented in the committee, thereby ensuring a proportional representation of all voters in the committee.
A committee is said to satisfy JR if in every sufficiently large group of voters who all agree on a set of candidates, at least one voter approves a candidate in the committee.

To find committees satisfying these axioms in practice, polynomial-time algorithms computing committees fulfilling one or even multiple of these axioms are sought.
For example, committees should satisfy a fairness and an efficiency axiom like JR and Pareto optimality in order to represent voters fairly while no other committee exists, which no voter finds worse and at least one voter prefers, thereby maximizing satisfaction with the committee in a sense.
The PAV voting rule is a rule for committees satisfying EJR+ and Pareto optimality \citep{pavpoeff,PropAxiomsForMultiwVoting}.
However, determining a winning committee by the PAV rule has been proven to be NP-hard in general \citep{pavnphard}.
Finding a polynomial-time algorithm for Pareto optimal committees satisfying EJR+ or even JR is an open problem \citep{peters_facets_2025}.

To alleviate the complexity of computing committees satisfying some axioms and analyzing properties of different axioms as well as to show stronger results, researchers proposed different restricted domains, that enforce more structure to voting instances by restricting the way voters can approve candidates.
Although allowing every voter to approve any candidates is the most general form of approval-based voting, one could argue that imposing restrictions on votes more accurately reflects real-world applications.
For example, when choosing among several options, most people tend to approve similar options rather than unrelated ones.
\cite{elkind2015structuredichotomouspreferences} proposed several restrictions formalizing this voting behavior and giving more structure to voting instances.
A voting instance satisfies the Candidate Interval property, if a total order can be imposed on all candidates, such that every voter approves an interval of the candidates within that ordering.
For example, though much simplified, political parties can be imagined to be ordered from left leaning to right leaning.
In this model, a voter approves similar candidates and therefore an interval of candidates in that ordering.
Another example of a restricted domain is the Voter Interval property, modeling a one-dimensional similarity of voters opinions and only voters of similar opinions approving each candidate, that, vice versa, is satisfied if a total order can be imposed on all voters, such that every candidate is approved by an interval of voters.

\subsection{Our Contribution}

To understand whether Pareto optimality is a suitable and practical to apply axiom on committees and to make finding polynomial-time algorithms to compute committees satisfying Pareto optimality and other axioms simultaneously, we develop a deeper understanding of the space of Pareto optimal committees in general and for two restricted domains.

Our analysis of Pareto optimal committees is centered on committee monotonicity and reconfigurability.
Committee monotonicity holds if, for every size-$k$ Pareto optimal committee, a candidate exists whose addition results in a size-$(k+1)$ Pareto optimal committee.
For reconfigurability, we define the Pareto optimality reconfiguration graph for committee size $k$ as a graph whose nodes represent size-$k$ Pareto optimal committees, with an edge between two nodes if their corresponding committees share exactly $k - 1$ candidates, analogously to \cite{reconfiguration-jr} for committees satisfying JR.

In \Cref{prt:restricted}, we present a simple characterization of Pareto optimal in specific restricted domains and draw conclusions about the structure of Pareto optimal committees when this characterization applies.

To this end, we introduce the \emph{Single Dominance Only (SDO)} property for voting instances, in \Cref{chp:sdo}.
We say that a voting instance satisfies SDO, if a size-$k$ committee is Pareto optimal, if no candidate in the committee is Pareto dominated by any candidate not in the committee.
In \Cref{sec:sdo-candidate-interval} and \Cref{sec:sdo-voter-interval}, we show that the SDO property holds for any voting that satisfies Candidate Interval or Voter Interval.
The characterization of Pareto optimality given by the SDO property implies that every union of Pareto optimal committees also is Pareto optimal, which is not true in general.

In \Cref{subsec:sdo-impl-monotonicity}, we use this to show that committee monotonicity holds for all Pareto optimal committees in any voting instance satisfying SDO.
As committee monotonicity implies an inherent structure of Pareto optimal committees, providing a sense of well-behavedness, we can conclude that Pareto optimality is a practical axiom of efficiency to work with in these restricted domains.

In particular, we go on to describe a polynomial-time algorithm that always returns a committee satisfying Pareto optimality and EJR+ simultaneously, given a voting instance satisfying SDO, in \Cref{subsec:sdo-impl-ejrp-algo}.
This algorithm is based on a known algorithm for committees satisfying EJR+ by \cite{PropAxiomsForMultiwVoting} and uses that committees satisfy committee monotonicity in the restricted domains considered.
While the PAV rule has been shown to be polynomial-time computable by \cite{PavInPolyTime} for the Candidate Interval domain, we thereby resolve the open question of finding fair and efficient committees for the Voter Interval domain in polynomial time.

Further, in \Cref{subsec:sdo-impl-reconfig-graph}, we show that the Pareto optimality reconfiguration graph is always connected and the distance of two size-$k$ committees $A$ and $B$ is $k - \lvert A \cap B \rvert$.
This implies that, given two Pareto optimal committees of the same size, one can always be transformed into the other by replacing candidates one at a time with candidates not in the committee, while retaining Pareto optimality at every step, and without using auxiliary candidates outside the two committees.
By analyzing the structure of the reconfiguration graph, we gain an understanding of how Pareto optimal committees relate to each other.

In \Cref{chp:counting-voter-interval}, we conjecture that the number of Pareto optimal committees for a voting instance satisfying Voter Interval can be counted in polynomial time.
We describe an algorithm using dynamic programming for counting the number of such committees and explain the intuition behind it.

In \Cref{prt:general-case}, we explain why our previous results for Pareto optimal committees do not translate to the general case.
In \Cref{chp:general-monotonicity}, we prove that committee monotonicity holds when either only a few candidates are in the committee or only a few candidates are not in the committee.
By explaining why extending this to larger numbers of candidates fails, we provide a better understanding of what makes it hard to prove committee monotonicity for Pareto optimal committees in general.

In \Cref{chp:general-reconfiguration}, we first show that a voting instance exists in which a pair of size-$2$ Pareto optimal committees has a distance of more than two, showing that our result from before does not hold.
We continue by proposing applying the \emph{canonical configuration} method \citep{introtoreconf}, a strategy commonly used in reconfiguration problems, for proving the connectedness of the Pareto optimality reconfiguration graph in the general case.
Our approach requires that, for every Pareto optimal committee that does not have a maximum approval score, a candidate exists, such that a candidate with a lower approval score can be replaced by that candidate while preserving Pareto optimality.
We show that one intuitive way of finding such a replacement does not work and leave the question unresolved.

\subsection{Related Work}
Previous research on Pareto optimality in general approval-based multiwinner voting as well as restricted domains has mostly been concerned with questions of computability.
To the best of our knowledge, there has been limited work discussing the structure of Pareto optimal committees.

\cite{Aziz2020} considered Pareto optimal in different voting methods and have shown that under various voting method, including approval-based multiwinner voting, checking whether a committee is Pareto optimal is coNP-complete, but did not consider structural properties of such committees.

Further work has been concerned with the Proportional Approval Voting (PAV) rule.
Winning committees according to PAV are Pareto optimal \citep{pavpoeff}.
\cite{JustifiedRepresentation} and \cite{PropAxiomsForMultiwVoting} have additionally shown that these committees also satisfy EJR and EJR+, respectively.
However, \cite{pavnphard} have shown that winner determination according to PAV is NP-hard, making it impractical for real world application.
To find more efficient algorithms for the PAV rule and other voting rules, \cite{elkind2015structuredichotomouspreferences} have proposed different restricted domains such as Candidate Interval and Voter Interval.
They provide algorithms for the PAV rule on these domains, though with running times exponential in certain parameters.
\cite{PavInPolyTime} further proved that winner determination according to the PAV rule is solvable in polynomial time for voting instances satisfying Voter Interval.

\cite{reconfiguration-jr} analyzed the reconfigurability of committees satisfying JR and EJR.
They showed that the reconfigurability graph for JR committees is not connected in general, but becomes connected when allowing ``2-approximate'' JR committees.
They further proved that winning committees under different voting rules such as PAV are all connected via JR committees.
Finally, they established that the reconfiguration graph for JR committees is connected for Candidate Interval and Voter Interval.

Outside of social choice theory, reconfiguration problems are studied in graph theory. A comprehensive survey of tools for reconfiguration problems was given by \cite{introtoreconf}.

\section{Preliminaries}
We define an \emph{approval-based multiwinner voting instance} by a set of $m$ \emph{candidates} $C = \{c_1, \dots, c_m\}$, a set of $n$ \emph{voters} $V = \{1, \dots, n\}$, for each voter $i$ a set $A_i \subseteq C$ of candidates they approve and a committee size $k$.
The set of all sets $A_i$ is called the \emph{approval profile} $\mathcal{A} = \{A_1, \dots, A_m\}$.
Thus, a voting instances is defined by a tuple $(C, V, \mathcal{A}, k)$.

A \emph{committee} is a subset of the candidates.
For a committee $W$ and a candidate $c$, we write $W-c = W \setminus \{c\}$ and $W+c = W \cup \{c\}$.

Further, for a voting instance $(C, V, \mathcal{A}, k)$, a \emph{winning committee} is a size-$k$ committee, consisting of all candidates elected as winners.
A \emph{voting rule} specifies which committees are valid winning committees given a voting instance.
E.g., a committee $W$ is a winning committee by the \emph{PAV} rule, if, for every other committee $W'$, $\sum_{i=1}^n H(\lvert A_i \cap W\rvert) \geq \sum_{i=1}^n H(\lvert A_i \cap W' \rvert)$, where $H(n) \coloneqq \sum_{i=1}^n \frac 1 n$ is the $n$-th Harmonic number, holds.

For every candidate $c \in C$ we define $V(c)$ as the set of voters approving candidate $c$ and call the cardinality of $V(c)$ the \emph{approval score} of $c$.
The approval score of a set of candidates $A \subseteq C$ is the sum of the approval scores of all candidates in $A$, i.e. $\sum_{c \in A}\left|V(c)\right|$.

\emph{Pareto optimality} for approval-based multiwinner voting is defined as a property of a committee as follows:

A size-$k$ committee $W \subseteq C$ is said to be \emph{dominated} by a size-$k'$ committee $D \subseteq C$, if for each voter $i \in V$ $\lvert A_i \cap W \rvert \leq \lvert A_i \cap D \rvert$ holds, and for at least one voter $v \in V$ $\lvert A_{v} \cap W\rvert < \lvert A_{v} \cap D \rvert$ holds.
We say, that committee $W$ is \emph{equal size dominated} or \emph{es-dominated} by the committee $D$ if $W$ is dominated by $D$ and $k = k'$ is true.
If $W$ is dominated by $D$, then we say that $D$ \emph{dominates} $W$.
A committee $W \subseteq C$ is \emph{Pareto optimal}, if no set of candidates $D \subseteq C$ es-dominates the committee $W$.

We further, say that candidate $c \in C$ is dominated by candidate $d \in C$, if the committee $\{c\}$ is dominated by the committee $\{d\}$.

A \emph{partial order} $\preceq \subseteq A \times A$ on a set $A$ is a reflexive, transitive and antisymmetric relation.
A partial order on a set $A$ is \emph{total} if for every $a, b \in A$ either $a \prec b$ or $b \prec a$ is true.
For a total order $\preceq \subseteq A \times A$ we write $a \prec b$ for $a, b \in A$, if $a \preceq b$ holds and $a \neq b$.

A subset $B \subseteq A$ is an \emph{interval} for a total order $\preceq \subseteq A \times A$, if there exists a partition into of set $A$ into subsets $L, R$ and $B$, such that $\ell \prec b \prec r$ holds for all candidates $\ell \in L$, $r \in R$ and $b \in B$.

We say that a voting instance $(C,V,\mathcal{A}, k)$ satisfies \emph{Voter Interval} \cite{elkind2015structuredichotomouspreferences}, if there is a total order $\preceq$ on the set of voters, such that $V(c)$ is an interval on $\preceq$ for every candidate $c \in C$.
For any two candidates $c_1, c_2\in C$ we say that the candidate $c_1$ is left of $c_2$ if neither $V(c_1)\subseteq V(c_2)$ nor $V(c_2)\subseteq V(c_1)$ holds, and there is a voter $i \in V(c_1)$ such that $i\prec j$ holds for all voters $j\in V(c_2)$.
We say that the candidate $c_1$ is right of the candidate $c_2$ if $c_2$ is left of $c_1$.
Further, for two voters $i, j \in V$, we say that voter $i$ is left of voter $j$, if $i \preceq j$, and right of $j$ if $j \preceq i$ holds.
For an example, see \Cref{fig:pre-voter-interval}.

Analogously, we say that a voting instance $(C,V,\mathcal{A}, k)$ satisfies \emph{Candidate Interval} \cite{elkind2015structuredichotomouspreferences}, if there is a total order $\preceq$ on the candidates, such that $A_i$ is an interval on $\preceq$ for every voter $i \in V$.
We say that a candidate $c_1 \in C$ is left of a candidate $c_2 \in C$ if $c_1 \preceq c_2$ holds, and right of $c_2$ if $c_2 \preceq c_1$ holds.
For an example, see \Cref{fig:pre-candidate-interval}.

\begin{figure}[h]
    \begin{minipage}[c]{0.4\linewidth}
    \centering
        \begin{tikzpicture}[yscale=0.5, xscale=0.65, voter/.style={anchor=south}]

        \foreach \i in {1,...,9}
        \node[voter] at (\i-0.5, -1.5) {$\i$};

            \drawcand{1}{1/1}
            \drawcand{2}{3/5}
            \drawcand[violet!30]{3}{3/3}
            \drawcand[blue!20]{4}{5/7}
            \drawcand[green!20]{5}{6/9}

        \end{tikzpicture}
        \caption{An example of Voter Interval.
         Here, the candidates $c_2$ and $c_3$ are neither left nor right of each other. The candidate $c_4$ is left of the candidate $c_5$ and right of all other candidates.
        }
        \label{fig:pre-voter-interval}
    \end{minipage}
    \hfill
    \begin{minipage}[c]{0.4\linewidth}
    \centering
        \begin{tikzpicture}[yscale=0.5, xscale=0.65, voter/.style={anchor=south}]

        \foreach \i in {1,...,9}
        \node[voter] at (\i-0.5, -1.5) {$c_\i$};

            \draw[fill=orange!20] (0,0) rectangle (1,1);
            \node at (0.5,0.5) {$1$};

            \draw[fill=red!40] (2,1) rectangle (5,2);
            \node at (2.5,1.5) {$2$};

            \draw[fill=violet!30] (2,2) rectangle (3,3);
            \node at (2.5,2.5) {$3$};

            \draw[fill=blue!20] (4,3) rectangle (7,4);
            \node at (4.5,3.5) {$4$};

            \draw[fill=green!20] (5,4) rectangle (9,5);
            \node at (5.5,4.5) {$5$};

        \end{tikzpicture}
        \caption{An example of Candidate Interval.
        Here, the candidate $c_2$ is right of the candidate $c_1$ and left of all other candidates.
        }
        \label{fig:pre-candidate-interval}
    \end{minipage}
\end{figure}

To study the space of Pareto optimal committees, we define the \emph{Pareto reconfiguration graph} $\Gamma_{Po}(\mathcal{A},k)$ for an approval profile $\mathcal{A}$ and a committee size $k$.
Each Pareto optimal size-$k$ Pareto optimal committee for approval profile $\mathcal{A}$ is a node in the graph $\Gamma_{Po}(\mathcal{A},k)$.
There is an edge between two nodes if the corresponding committees $W_1, W_2$ differ in only one candidate, i.e., if $|W_1 \cap W_2| = k-1$ holds.
We define the \emph{distance} between two committees $W_1, W_2$, denoted as $d(W_1, W_2)$, as the number of edges on the shortest path between $W_1$ and $W_2$ in $\Gamma_{Po}(\mathcal{A}, k)$.

\section{Restricted Domains}
\label{prt:restricted}

\subsection{Characterization of Pareto optimality}
\label{chp:sdo}
\cite{Aziz2020} have shown that verifying, whether a size-$k$ committee is Pareto optimal, is coNP-complete.
To check whether a size-$k$ committee in a voting instance with $n$ candidates is Pareto optimal it has to be compared to all $\binom n k - 1$ other committees by the definition of Pareto optimality.
An equivalent definition of Pareto optimality is that a committee $A$ is Pareto optimal, if no subset $A' \subseteq A$ of the committee is es-dominated by a disjoint, same size set of candidates $B$.
However, as the number of size-$j$ subsets of a size-$k$ committee is $\binom k j$, checking whether a committee is Pareto optimal by comparing all possible subsets is not feasible and thus also makes it difficult to prove and analyze properties of Pareto optimality as well.
We therefore first consider a simpler characterization of Pareto optimality by introducing the \emph{Single Dominance Only} (SDO) property, that limits Pareto optimality to only having to check whether any single candidate of the committee is dominated by itself.

\begin{definition}
    We say that a voting instance $(C, V,\mathcal{A},k)$ satisfies the \emph{Single Dominance Only (SDO)} property, if the following characterization holds:

    A committee $W \subseteq C$ is Pareto optimal, iff for every candidate $c \in W$ no candidate $d \in C \setminus W$ dominates $c$.
\end{definition}

We show that the SDO property, though much simpler than the definition of Pareto optimality, is applicable in practice by proving that this characterization of Pareto optimality holds for all voting instances satisfying Candidate Interval or Voter Interval.

\subsubsection{Candidate Interval}
\label{sec:sdo-candidate-interval}

The Candidate Interval domain on voting instances as proposed by \cite{elkind2015structuredichotomouspreferences} models a one-dimensional similarity of candidates and voters approving only similar candidates and thereby gives more structure to how voters approve candidates.
We make use of the structure enforced on the approval profile by this restricted domain to show that any Candidate Interval voting instance also satisfies SDO.

We prove that the SDO property holds for any such voting instance by way of contradiction by assuming a size-$k$ Pareto optimal committee $A$, for which no candidate is dominated, exists, such that another size-$k$ Pareto optimal committee $B$ dominates committee $A$.
We then show that committee $B$ would have to have at least $k+1$ candidates for it to dominate committee $A$ contradicting the choice of committee $B$.
We first want to explain why that is on a higher level.
Every voter approves at least as many candidates in committee $B$ as in $A$.
Thus, for each candidate $c \in A$, each voter $v \in V(c)$ approving candidate $c$ also approves a candidate $d_v \in B$.
We show that these candidates $d_v$ must be on the left and the right of candidate $c$ in the total order of candidates induced by the Candidate Interval property, not on only left of candidate $c$ or right of candidate $c$, for every such voter voter $v \in V(C)$.
Thus, every candidate $c \in A$ requires two candidates $d_1, d_2 \in B$ to exist, such that candidate $d_1$ is left of $c$ and candidate $d_2$ right of $c$.
We further show, that each candidate $d \in B$ cannot be the left required candidate of two candidates $c_1, c_2 \in A$ or the right required candidate of two candidates $c_1, c_2 \in A$ simultaneously, by which we can conclude that at least $k+1$ candidates are in committee $B$ contradicting the choice of $B$.

To show our claim, we first prove the following lemma.

\begin{lemma}
    Let $(C,V,\mathcal{A},k)$ be an approval-based multiwinner voting instance that satisfies the candidate interval property for a total order $\preceq$ on the set of candidates.

    If a candidate $c \in C$ is dominated by a set of candidates $D \subseteq C$ and $V(c) \nsubseteq V(d)$ holds for every candidate $d \in D$, then there are candidates $d_1, d_2 \in D$ such that $d_1 \prec c \prec d_2$ holds.
    \label{lem:restricted-candidate-interval-left-and-right}
\end{lemma}
\begin{proof}
    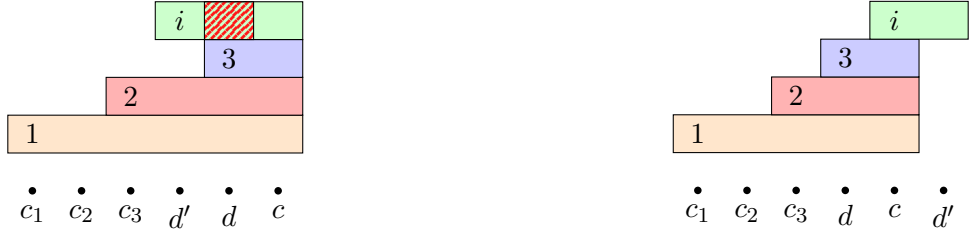
\begin{figure}[h]
        \centering
        \begin{minipage}{0.45\linewidth}
            \centering
            \begin{tikzpicture}[yscale=0.5, xscale=0.65, voter/.style={anchor=south}]
            \foreach \i in {1,...,3}
                \node[circle, fill=black, inner sep=1pt, label=below:$c_\i$] at (\i-0.5,-1) {};
            \node[circle, fill=black, inner sep=1pt, label=below:$d'$] at (3.5,-1) {};
            \node[circle, fill=black, inner sep=1pt, label=below:$d$] at (4.5,-1) {};
            \node[circle, fill=black, inner sep=1pt, label=below:$c$] at (5.5,-1) {};

            \draw[fill=orange!20] (0,0) rectangle (6,1);
            \node at (0.5,0.5) {$1$};

            \draw[fill=red!30] (2,1) rectangle (6,2);
            \node at (2.5,1.5) {$2$};

            \draw[fill=blue!20] (4,2) rectangle (6,3);
            \node at (4.5,2.5) {$3$};

            \draw[fill=green!20] (3,3) rectangle (6,4);
            \node at (3.5,3.5) {$i$};

            \draw[pattern={Lines[angle=45,distance=2pt,line width=1pt]},pattern color=red!90] (4,3) rectangle (5,4);
            \end{tikzpicture}
        \end{minipage}
        \hfill
        \begin{minipage}{0.45\linewidth}
            \centering
            \begin{tikzpicture}[yscale=0.5, xscale=0.65, voter/.style={anchor=south}]
            \foreach \i in {1,...,3}
                \node[circle, fill=black, inner sep=1pt, label=below:$c_\i$] at (\i-0.5,-1) {};
            \node[circle, fill=black, inner sep=1pt, label=below:$d$] at (3.5,-1) {};
            \node[circle, fill=black, inner sep=1pt, label=below:$c$] at (4.5,-1) {};
            \node[circle, fill=black, inner sep=1pt, label=below:$d'$] at (5.5,-1) {};

            \draw[fill=orange!20] (0,0) rectangle (5,1);
            \node at (0.5,0.5) {$1$};

            \draw[fill=red!30] (2,1) rectangle (5,2);
            \node at (2.5,1.5) {$2$};

            \draw[fill=blue!20] (3,2) rectangle (5,3);
            \node at (3.5,2.5) {$3$};

            \draw[fill=green!20] (4,3) rectangle (6,4);
            \node at (4.5,3.5) {$i$};
            \end{tikzpicture}
        \end{minipage}
        \caption{Here, candidate $c$ is the candidate dominated by a set of candidates, candidate $d$ is the rightmost candidate of. Voter $i$ approves candidate $c$ but not $d$ and thus must approve another candidate $d' \in D$. If candidate $d'$ was left of $d$, the Candidate Interval property would be violated. Thus, candidate $d'$ must be right of candidate $d$, contradicting the choice of $d$ as the rightmost candidate in the dominating set.}
        \label{fig:restricted-candidate-interval-left-right}
    \end{figure}

    Let $c \in C$ be a candidate dominated by a subset of candidates $D \subseteq C$ such that $V(c) \subseteq V(d)$ does not hold for any $d \in D$.
    For the sake of contradiction, assume that all candidates in $D$ are left of $c$.
    Let $d \in D$ be the rightmost candidate of all candidates in $D$.
    We now show that a voter exists, who approves candidate $c$ but not $d$.
    As the candidate $c$ is dominated by the candidates in $D$, that voter must also approve a candidate in $D$ and hence approves a candidate right of $c$ that is in $D$, as is depicted in \Cref{fig:restricted-candidate-interval-left-right}.

    As we assumed that all candidates in $D$ are left of candidate $c$, candidate $d$ is left of candidate $c$.
    Additionally, a voter $i \in V$ approving candidate $c$ but not candidate $d$ exists, as we assumed that $V(c) \nsubseteq V(d)$ holds.
    Because the set of candidates $D$ dominates $c$, a candidate $d' \in D$ who is approved by voter $i$ exists.
    Further, candidate $d'$ must be left of $d$, as $d$ was defined as the rightmost candidate among all candidates in $D$.
    By the candidate interval property, the set of candidates voter $i$ approves must be an interval on $\preceq$.
    Thus, as $d' \prec d \prec c$ holds, voter $i$ must also approve candidate $d$, contradicting the choice of voter $i$.

    Hence, not all candidates $d \in D$ are left of candidate $c$. Analogously, not all candidates $d \in D$ are right of candidate $c$.
    Therefore, two candidates $d_1, d_2 \in D$, such that $d_1 \prec c \prec d_2$ holds, exist.
\end{proof}

We can now show that any voting instance satisfying Candidate Interval voting also fulfills SDO.

\begin{theorem}
    Let $(C,V,\mathcal{A},k)$ be an approval-based multiwinner voting instance that satisfies the candidate interval property for some total order $\preceq$ on the set of candidates.

    Then $(C, V, \mathcal{A},k)$ also satisfies the SDO property.
    \label{thm:restricted-candidate-interval-characterization-po}
\end{theorem}

\begin{proof}
    We first show that, if a committee $W \subseteq C$ is Pareto optimal, then no candidate $c \in W$ is dominated by a candidate $d \in C \setminus W$.

    For the sake of contradiction, assume there is a candidate $c \in W$ who is dominated by a candidate $d \in C \setminus W$.
    Hence, the committee $W$ is dominated by the committee $W-c+d$ and thus not Pareto optimal, contradicting our choice of the committee $W$.

    It remains to show that any committee $W \subseteq C$, for which candidate $c \in W$ is dominated by a candidate $d \in C \setminus W$, is Pareto optimal.

    Let $W \subseteq C$ be a committee such that for all candidates $c \in W$ no candidate $d \in C \setminus W$ dominates $c$.
    For the sake of contradiction, assume $W$ is not Pareto optimal.

    Let $S \subseteq C \setminus W$ be a subset of the remaining committees with minimum cardinality, such that $S$ dominates a subset of the committee $W' \subseteq W$ with $\lvert S\rvert \leq \lvert W' \rvert$.
    Further, let $U \subseteq W$ be a minimal subset of the committee that is dominated by $S$ and for which $\lvert S \rvert \leq \lvert U \rvert$ holds.

    We show that no such subsets $S$ and $U$ can exist.

    We first make a case distinction by the cardinality of $U$ and $S$.

    Case 1: Assume $\lvert S \rvert < \lvert U \rvert$ holds.

    As the subset $S$ dominates $U$, it also holds that $S$ dominates all subsets of $U$.
    In particular that means, that some subset $U' \subset U$ with $\lvert S \rvert \leq \lvert U'\rvert$ exists, contradicting our choice of $U$.

    Case 2: Assume $\lvert S \rvert = \lvert U \rvert$ holds.

    Let $c \in U$ be the rightmost candidate in $U$.
    By \Cref{lem:restricted-candidate-interval-left-and-right} a candidate $d \in S$ that is right of $c$ must exist.
    Further, we show that some voter $j \in V$ approving candidate $d$ and a candidate $c' \in U-c$ exists.
    For the sake of contradiction, assume such a voter did not exist, i.e.\ no voter approving $d$ approves any candidate in $U-c$.
    Hence, $S-d$ dominates $U-c$ and, as $\lvert S\rvert = \lvert U\rvert$ holds, $\lvert S-d\rvert = \lvert U-c\rvert$, contradicting the choice of $S$ and $U$.
    Thus, such a voter $j \in V$ exists.

    To continue, we make another case distinction.

    Case 2.1: Assume a voter $j \in V$, approving candidate $d$, who approves as many candidates in $U-c$ as in $S$ exists.

    Thus, voter $j$ approves a candidate $c' \in U-c$.
    By the choice of candidates $c$ and $d$, we know that $c' \prec c \prec d$ holds.
    This yields that voter $j$ approves more candidates in $U$ than in $S$, contradicting our initial choice of $U$ and $S$.

    Case 2.2: Assume every voter $j \in V$ approving $d$ approves less candidates in $U-c$ than in $S-d$.

    By our choice of $U$ and $S$ the subset $U-c$ is not dominated by the subset $S-d$.
    This yields that either all voters not approving candidate $d$ approve as many candidates in $U-c$ as in $S-d$ or a voter approves more candidates in $S-d$ than in $U-c$.
    The latter cannot be true, as $S$ dominates $U$ and every voter approving $d$ approves less candidates in $U-c$ than in $S-d$.
    Hence, all voters not approving $d$ approve as many candidates in $U-c$ as in $S-d$.

    Let $j' \in V$ be a voter approving $c$ but not $d$.
    As shown above, voter $j'$ approves as many candidates in $U-c$ as in $S-d$ and, thus, one more candidate in $U$ than in $S-d$.
    By our choice of $U$ and $S$, the set $U$ is dominated by $S$, yielding that $j'$ approves at least as many candidates in $S$ and $U$.
    Thus, $j'$ approves as many candidates in $U-c$ as in $S-d$ and approves $d$, contradicting our assumption that every voter approving $d$ approves less candidates in $U-c$ than in $S-d$.

    We can conclude that neither $\lvert S \rvert < \lvert U \rvert$ nor $\lvert S \rvert = \lvert U \rvert$ holds, and thus $\lvert S\rvert > \lvert U\rvert$ holds, contradicting our choice of sets $S$ and $U$.

    Hence, the committee $W$ is Pareto optimal.
\end{proof}

Thus, all voting instances satisfying Candidate Interval also satisfy SDO.
We can hence conclude that every property shown to hold for voting instances satisfying SDO also hold for voting instances satisfying Candidate Interval

\subsubsection{Voter Interval}
\label{sec:sdo-voter-interval}

We next consider voting instances satisfying Voter Interval, another restricted domain proposed by \cite{elkind2015structuredichotomouspreferences}.
Voter Interval is defined analogously to Candidate Interval with a total order on all voters instead of candidates and candidates being approved by an interval of voters.

For Voter Interval, as for Candidate Interval, it also holds that any voting instance satisfies SDO as well, yielding two applications of our SDO property.

In the following proof, we show our claim by way of contradiction.

First, we want to give an intuition for why the SDO property holds for voting instances satisfying Voter Interval:

For a subset of the committee to be dominated, every voter has to approve at least as many candidates in a dominating set as in the subset of the committee.
Thus, for each candidate in the committee, the interval of voters approving that candidate has to also approve candidates in the dominating set.
Because no candidate in the committee is dominated, none of the candidates in the dominating set are approved by the entire interval of voters.
Instead, the interval of voters is split into at least two intervals of voters approving different candidates, leading to the dominating subset to always be larger than the subset of the committee dominated.

\begin{theorem}
    Let $(C,V,\mathcal{A},k)$ be an approval-based multiwinner voting instance satisfying Voter Interval for a total order $\preceq$ on the set of voters.

    Then $(C, V, \mathcal{A},k)$ also satisfies the SDO property.
    \label{thm:restricted-voter-interval-characterization-po}
\end{theorem}
\begin{proof}
    We first show that if a committee $W \subseteq C$ is Pareto optimal, then no candidate $c \in W$ is dominated by a candidate $d \in C \setminus W$.

    For the sake of contradiction, assume there is a candidate $c \in W$ that is dominated by a candidate $d \in C \setminus W$.
    Hence, the committee $W$ is dominated by the committee $W-c+d$ and thus not Pareto optimal, contradicting our choice of the committee $W$.

    It remains to show that any committee $W \subseteq C$, for which no candidate $c \in W$ is dominated by a candidate $d \in C \setminus W$, is Pareto optimal.

    Let $W \subseteq C$ be a committee such that no candidate $c \in W$ is dominated by any candidate $d \in C\setminus W$.

    For the sake of contradiction, assume that $W$ is not Pareto optimal.

    Let $S \subseteq C \setminus W$ be a subset of the remaining committee with minimum cardinality, such that $S$ dominates a subset of the committee $W' \subseteq W$ with $\lvert S\rvert \leq \lvert W' \rvert$.
    Further, let $U \subseteq W$ be a subset of the committee that is dominated by $S$ with minimum cardinality, for which $\lvert S \rvert \leq \lvert U \rvert$ holds, i.e., $U$ is the smallest committee subset dominated by $S$ that is at least as large as $S$.

    We show that such subsets $S$ and $U$ cannot exist by way of contradiction.

    Let $i \in V$ be the rightmost voter approving a candidate in $U$ and $c \in U$ be a candidate approved by voter $i$.

    Every voter approves at least as many candidates in $S$ as in $U$, as $S$ dominates $U$.
    Hence, voter $i$ also approves a candidate $d \in S$.

    Further, let $\ell \in V$ be the leftmost voter approving candidate $c$.

    To continue we make a case distinction.

    \begin{figure}[h]
        \begin{minipage}[c]{0.4\linewidth}
        \centering
        \begin{tikzpicture}[yscale=0.5, xscale=0.65]
            \node[circle, fill=black, inner sep=1pt, label=below:$\dots$] (voters1) at (1.5,-1) {};
            \node[circle, fill=black, inner sep=1pt, label=below:$\ell$]
            (l) at (2.5,-1) {};
            \node[circle, fill=black, inner sep=1pt, label=below:$\dots$]
            (voters2) at (3.5,-1) {};
            \node[circle, fill=black, inner sep=1pt, label=below:$i$]
            (i) at (4.5,-1) {};
            \node[circle, fill=black, inner sep=1pt, label=below:$j$]
            (j) at (5.5,-1) {};
            \node[circle, fill=black, inner sep=1pt, label=below:$\dots$]
            (voters3) at (6.5,-1) {};

            \draw[fill=orange!20] (2,0) rectangle (5,1);
            \node at (2.5,0.5) {$c$};

            \draw[fill=cyan!20] (1.5,1) rectangle (5,2);
            \node at (2,1.5) {$d$};
        \end{tikzpicture}
        \caption{Candidate $d$ has $i$ as their rightmost voter and $\ell$ approves $d$}
        \label{fig:voter-interval-cases-1-1}
        \end{minipage}
        \hfill
        \begin{minipage}[c]{0.4\linewidth}
        \centering
            \begin{tikzpicture}[yscale=0.5, xscale=0.65]
                \node[circle, fill=black, inner sep=1pt, label=below:$\dots$] (voters1) at (1.5,-1) {};
            \node[circle, fill=black, inner sep=1pt, label=below:$\ell$]
            (l) at (2.5,-1) {};
            \node[circle, fill=black, inner sep=1pt, label=below:$\dots$]
            (voters2) at (3.5,-1) {};
            \node[circle, fill=black, inner sep=1pt, label=below:$i$]
            (i) at (4.5,-1) {};
            \node[circle, fill=black, inner sep=1pt, label=below:$j$]
            (j) at (5.5,-1) {};
            \node[circle, fill=black, inner sep=1pt, label=below:$\dots$]
            (voters3) at (6.5,-1) {};

            \draw[fill=orange!20] (2,0) rectangle (5,1);
            \node at (2.5,0.5) {$c$};

            \draw[fill=cyan!20] (3.5,1) rectangle (5,2);
            \node at (4,1.5) {$d$};
        \end{tikzpicture}
        \caption{Candidate $d$ has $i$ as their rightmost voter and $\ell$ does not approve $d$}
        \label{fig:voter-interval-cases-1-2}
        \end{minipage}
        \vfill
        \begin{minipage}[c]{0.4\linewidth}
        \centering
        \begin{tikzpicture}[yscale=0.5, xscale=0.65]
            \node[circle, fill=black, inner sep=1pt, label=below:$\dots$] (voters1) at (1.5,-1) {};
            \node[circle, fill=black, inner sep=1pt, label=below:$\ell$]
            (l) at (2.5,-1) {};
            \node[circle, fill=black, inner sep=1pt, label=below:$\dots$]
            (voters2) at (3.5,-1) {};
            \node[circle, fill=black, inner sep=1pt, label=below:$i$]
            (i) at (4.5,-1) {};
            \node[circle, fill=black, inner sep=1pt, label=below:$j$]
            (j) at (5.5,-1) {};
            \node[circle, fill=black, inner sep=1pt, label=below:$\dots$]
            (voters3) at (6.5,-1) {};

            \draw[fill=orange!20] (2,0) rectangle (5,1);
            \node at (2.5,0.5) {$c$};

            \draw[fill=cyan!20] (1.5,1) rectangle (6.5,2);
            \node at (2,1.5) {$d$};
        \end{tikzpicture}
        \caption{Candidate $d$ is approved by voter $j$ and $\ell$ approves $d$}
        \label{fig:voter-interval-cases-2-1}
        \end{minipage}
        \hfill
        \begin{minipage}[c]{0.4\linewidth}
        \centering
            \begin{tikzpicture}[yscale=0.5, xscale=0.65]
            \node[circle, fill=black, inner sep=1pt, label=below:$\dots$] (voters1) at (1.5,-1) {};
            \node[circle, fill=black, inner sep=1pt, label=below:$\ell$]
            (l) at (2.5,-1) {};
            \node[circle, fill=black, inner sep=1pt, label=below:$\dots$]
            (voters2) at (3.5,-1) {};
            \node[circle, fill=black, inner sep=1pt, label=below:$\ell'$]
            (lprime) at (4.5,-1) {};
            \node[circle, fill=black, inner sep=1pt, label=below:$\dots$]
            (voters3) at (5.5,-1) {};
            \node[circle, fill=black, inner sep=1pt, label=below:$i$]
            (i) at (6.5,-1) {};voter
            \node[circle, fill=black, inner sep=1pt, label=below:$j$]
            (j) at (7.5,-1) {};
            \node[circle, fill=black, inner sep=1pt, label=below:$\dots$]
            (voters4) at (8.5,-1) {};

            \draw[fill=orange!20] (2,0) rectangle (7,1);
            \node at (2.5,0.5) {$c$};

            \draw[fill=cyan!20] (4,1) rectangle (8.5,2);
            \node at (4.5,1.5) {$d$};
        \end{tikzpicture}
        \caption{Candidate $d$ is approved by voter $j$ and $\ell$ does not approve $d$}
        \label{fig:voter-interval-cases-2-2}
        \end{minipage}
    \end{figure}

    Case 1: Assume voter $i$ is the rightmost voter approving candidate $d$.

    Then one of the following cases holds:

    Case 1.1: Voter $\ell$ approves candidate $d$, as depicted in \Cref{fig:voter-interval-cases-1-1}.
    Voter $\ell$ is either the left most voter approving $d$ which implies $V(c) = V(d)$ or some voter left of $\ell$ also approves $d$.
    If $V(c) = V(d)$ holds, then the set $S-d$ dominates $U-c$ contradicting our choice of $U$ and $S$.
    Thus, a voter left of $\ell$ approves candidate $d$ but not $c$.

    For every voter $v \in V$ approving candidate $c$, it holds that $\ell \preceq v \preceq  i$.
    By the voter interval property, every voter $v \in V$ approving candidate $c$ also approves $d$.
    Hence, candidate $d$ dominates candidate $c$, contradicting our assumption of no candidate in $U$ being dominated by a candidate in $S$.

    Case 1.2: Voter $\ell$ does not approve candidate $d$ as in \Cref{fig:voter-interval-cases-1-2}.

    Thus, for all voters $v \in V$ approving candidate $d$, $\ell \prec v \preceq i$ holds and, hence, by the voter interval property, voter $v$ also approves candidate $c$.
    This yields that candidate $d$ is dominated by $c$, contradicting our choice of $S$ and $U$ as $S-d$ therefore dominates $U-c$ and $\lvert S \rvert- \lvert U \rvert = \lvert S-d \rvert - \lvert U - c \rvert$ but $S$ was chosen to be of minimum size.

    Case 2:
    Let $j \in V$ be the voter to the right of voter $i$.

    Assume voter $j$ approves $d$.

    Again, we distinguish two further cases.

    Case 2.1: Assume voter $\ell$ approves candidate $d$, as depicted in \Cref{fig:voter-interval-cases-2-1}.
    Voter $\ell$ is either the leftmost candidate approving $d$ or a voter left of $\ell$ also approves $d$.

    Analogously to case 1.1, we can conclude that candidate $d$ dominates candidate $c$, contradicting our assumption.

    Case 2.2: Assume voter $\ell$ does not approve candidate $d$, as depicted in \Cref{fig:voter-interval-cases-2-2}.

    Let $\ell' \in V$ be the rightmost voter approving candidate $c$ but not candidate $d$.

    As the subset of candidates $S$ dominates $U$, every voter $v \in V$ approves at most as many candidates in $U$ as in $S$.

    In particular, each voter $v \in V$ with $\ell \preceq v \preceq \ell'$ approves at most as many candidates in $U$ as in $S$.
    Additionally, voter $v$ approves candidate $c$ but does not approve candidate $d$.

    Hence, each voter $v \in V$ with $\ell \preceq v \preceq \ell'$ approves strictly less candidates in $U-c$ than in $S-d$.

    Further, each voter $v \in V$ with $\ell'\prec v \preceq i$ approves at most as many candidates in $U-c$ as in $S-d$ as voter $v$ approves exactly one more candidate in $U$ than in $U-c$ and exactly one more candidate in $S$ than in $S-d$.

    All voters $v \in V$ with $i \prec v$ do not approve any candidates in $U$ as voter $i$ is the rightmost voter approving a candidate in $U$.

    Additionally, all voters $v \in V$ with $v \prec \ell$ neither approve candidate $c$ nor $d$.

    We can thereby conclude that $S-d$ dominates $U-c$, contradicting our initial choice of $S$ and $U$.

    This yields that no such subsets $S$ and $U$ can exist.

    Hence, the committee $W$ is Pareto optimal.
\end{proof}

Thus, all voting instances satisfying Voter Interval satisfy SDO as well.

\subsubsection{Implications of SDO}

We have established that every voting instance satisfying Candidate Interval or Voter Interval also satisfy SDO.
We now show some structural properties of Pareto optimal committees for SDO voting instances, which can thereby be transferred to the Candidate Interval and Voter Interval domains.

\paragraph{Monotonicity}
\label{subsec:sdo-impl-monotonicity}

Committee monotonicity directly follows from the SDO property.
The SDO property allows us to check whether a committee is Pareto optimal by checking whether each subset of a partition of the committee is Pareto optimal.

\begin{lemma}
    Let $(C,V, \mathcal{A}, k)$ be a voting instance satisfying SDO and $A \subseteq C$ and $B \subseteq C$ two disjoint Pareto optimal committees of any size.
    The committee $A \cup B$ is also Pareto optimal.
    \label{lem:sdo-po-union}
\end{lemma}
\begin{proof}
    Let $A, B \subseteq C$ be such committees.
    As the SDO property holds, no candidate $c \in A$ is dominated by any candidate $d \in C \setminus A$ and no candidate $c' \in B$ is dominated by any candidate $d' \in C \setminus$.
    Thus, no candidate $c^* \in A \cup B$ is dominated by any candidate $d^* \in C \setminus (A \cup B)$.
    Hence, the committee $A \cup B$ is also Pareto optimal by the SDO property.
\end{proof}

Thus, to prove that, for every Pareto optimal size-$k$ committee $W$, a candidate $c \in C \setminus W$ exists, such that $W+c$ is a size-$(k+1)$ Pareto optimal committee, it suffices to show that a candidate $c \in C \setminus W$ who is not dominated by any other candidate $d \in C \setminus W$ exists.

\begin{lemma}
    Let $(C,V,\mathcal{A}, k)$ be a voting instance for which the SDO property holds.

    Let $k<|C|$. For any size-$k$ Pareto optimal committee $W \subseteq C$, some candidate $c\in C \setminus W$ exists, such that the committee $W+c$ is also Pareto optimal with respect to $k+1$.
    \label{lem:restricted-monotonicity}
\end{lemma}
\begin{proof}
    Let $W \subseteq C$ be a size-$k$ Pareto optimal committee.
    Thus, no candidate $c \in W$ is dominated by any candidate $d \in C \setminus W$, by the SDO property.

    Further, let $c \in C\setminus W$ be a candidate with maximum approval score among all candidates in $C \setminus W$.
    Then candidate $c$ is not dominated by any candidate $d \in C \setminus W$ as not every voter approving candidate $c$ can also approve candidate $d$ because the approval score of candidate $c$ is maximum.
    Hence, $W+c$ is a size-$k+1$ Pareto optimal committee, by \Cref{lem:sdo-po-union}.
\end{proof}

The SDO property is essential here, as, if SDO did not hold, it would be necessary to show that a candidate exists that is not es-dominated together with any subset of the original committee.

It now follows that committee monotonicity for Pareto optimality holds for every voting instance satisfying Candidate Interval or Voter Interval.

\begin{corollary}
    Let $(C,V,\mathcal{A}, k)$ be a voting instance satisfying Candidate Interval.

    Let $k<|C|$. For any size-$k$ Pareto optimal committee $W \subseteq C$, a candidate $d\in C \setminus W$ exists, such that the committee $W+d$ is also Pareto optimal with respect to $k+1$.
    \label{cor:restricted-candidate-interval-monotonicity}
\end{corollary}
\begin{proof}
    Let $(C, V, \mathcal{A}, k)$ be a voting instances satisfying Candidate Interval.
    By \Cref{thm:restricted-candidate-interval-characterization-po}, the SDO property holds for this voting instance.

    Hence, by \Cref{lem:restricted-monotonicity}, for any size-$k$ Pareto optimal committee $W \subseteq C$, a candidate $d\in C \setminus W$ exists, such that the committee $W+d$ is also Pareto optimal with respect to $k+1$.
\end{proof}

\begin{corollary}
    Let $(C,V,\mathcal{A}, k)$ be a voting instance satisfying Voter Interval.

    Let $k<|C|$. For any size-$k$ Pareto optimal committee $W \subseteq C$, a candidate $d\in C \setminus W$ exists, such that the committee $W+d$ is also Pareto optimal with respect to $k+1$.
    \label{cor:restricted-voter-interval-monotonicity}
\end{corollary}
\begin{proof}
    Let $(C, V, \mathcal{A}, k)$ be a voting instances satisfying Voter Interval.
    By \Cref{thm:restricted-voter-interval-characterization-po}, the SDO property holds for this voting instance.

    Hence, by \Cref{lem:restricted-monotonicity}, for any size-$k$ Pareto optimal committee $W \subseteq C$, a candidate $d\in C \setminus W$ exists, such that the committee $W+d$ is also Pareto optimal with respect to $k+1$.
\end{proof}

Committee monotonicity is very useful when constructing committees, as, thus, Pareto optimal committees can be constructed step by step.
As we further characterized, by what candidates a Pareto optimal committee can be extended, this opens up the possibility to construct Pareto optimal committees with further properties by one by one extending a Pareto optimal committee by candidates satisfying further properties.

\paragraph{Constructing a Pareto optimal committee satisfying EJR+}
\label{subsec:sdo-impl-ejrp-algo}

We next want to use this structure of Pareto optimal committees for constructing a committee satisfying EJR+ and Pareto optimality simultaneously.
\cite{PropAxiomsForMultiwVoting} have described a simple greedy algorithm to construct a committee of at most $k$ candidates that satisfies EJR+, given a voting instance.
We want to extend that algorithm to construct a size-$k$ committee that satisfies both EJR+ and Pareto optimality given a voting instance $(C,V,\mathcal{A}, k)$ that satisfies SDO.

\begin{algorithm}
\caption{Greedy EJR+ \& PO committee for voting instances satisfying SDO}
\label{alg:greedy-ejr-po}
    $W \gets \emptyset$\;
    \For{$\ell$ in $k, \dots, 1$}{
    \While{there is $c\notin W$: $\lvert \{i \in V(c) \colon \lvert A_i \cap W \rvert < \ell\} \rvert \ge \frac{\ell n}{k}$ }{
    $N \gets \{c \notin W \mid \lvert \{i \in V(c) \colon \lvert A_i \cap W \rvert < \ell\}\rvert \geq \frac{\ell n}{k}\}$\;
    select $d \in \arg \max_{d \in N} \lvert V(d) \rvert$\;
    \label{algline:greedy-ejr-po-ejrpart-order}
    $W \gets W \cup \{d\}$\;
    }
    }
    \label{algline:greedy-ejr-po-after-forloop}
    \While{$\lvert W \rvert < k$}{
    select $c \in \arg \max_{c \notin W} \lvert V(c) \rvert$\;
    $W \gets W \cup \{c\}$\;
    }
    \Return{ $W$}\;
\end{algorithm}

\begin{theorem}
    Given a voting instance $(C,V,\mathcal{A}, k)$ satisfying SDO, \Cref{alg:greedy-ejr-po} returns a size-$k$ committee that satisfies both EJR+ and Pareto optimality.
    \label{thm:sdo-ejrp-po-poly-algo}
\end{theorem}
\begin{proof}
    We first show that \Cref{alg:greedy-ejr-po} returns a committee satisfying EJR+.
    The for-loop of our \Cref{alg:greedy-ejr-po} differs to the algorithm described by \cite{PropAxiomsForMultiwVoting} only by line~\ref{algline:greedy-ejr-po-ejrpart-order}.
    In the for-loop, candidates violating EJR+ are added to the committee in the original algorithm.
    We differ from this original algorithm only by deciding which exact candidate still violating EJR+ to add to the committee instead of adding any violating candidate.
    This does not impact the correctness of the algorithm.
    Thus, the committee $W$ satisfies EJR+ after the for-loop in line~\ref{algline:greedy-ejr-po-after-forloop}.

    In the remaining part of our algorithm no candidates are being removed, thus preserving EJR+, as \cite{PropAxiomsForMultiwVoting} have shown that the remaining candidates can be chosen arbitrarily in order for the final committee to satisfy EJR+ with respect to committee size $k$.

    We now show that \Cref{alg:greedy-ejr-po} returns a Pareto optimal committee.

    First, we show that every candidate selected in the for-loop is not dominated by any candidate not in the final committee returned.

    Any candidate $c$ that is selected in the for loop has a maximum approval score among all candidates in $N$.
    Hence, no candidate $d \in N$ dominates $d$.

    Let $d \in (C \setminus W)\setminus N$ be a candidate neither in the committee nor among the candidates in $N$.
    In particular, this means that $\lvert \{i \in V(d) \colon \lvert A_i \cap W \rvert < \ell\} \rvert \ge \frac{\ell n}{k}$ does not hold for candidate $d$.
    For the sake of contradiction, assume that $d$ dominates candidate $c$.
    Thus, every voter approving candidate $c$ also approves candidate $d$ by the definition of Pareto optimality.

    \noindent
    Therefore, $ \{i \in V(c) \colon \lvert A_i \cap W \rvert < \ell\} \subseteq \{i \in V(d) \colon \lvert A_i \cap W \rvert < \ell\}$ holds.

    \noindent
    As candidate $c$ was added to $W$ in the for-loop, $c \in N$ holds, and further \linebreak[4]$\lvert\{i \in V(c) \colon \lvert A_i \cap W \rvert < \ell\}\rvert \geq \frac{\ell n}{k}$ holds.

    \noindent
    Hence, $\lvert \{i \in V(d) \colon \lvert A_i \cap W \rvert < \ell\}\rvert \geq \lvert \{i \in V(c) \colon \lvert A_i \cap W \rvert < \ell\} \rvert \geq \frac{\ell n}{k}$ holds, contradicting our choice of candidate $d$.

    This yields, that no candidate selected to be added to $W$ in the for-loop is dominated by any candidate in $C \setminus W$.

    In the last part of our algorithm, every candidate added to $W$ has a maximum approval score among all candidates not in $W$ yet, and thereby is not dominated by any remaining candidate $d \in C \setminus W$.

    Hence, in the final committee $W$ that is being returned, no candidate $c \in W$ is dominated by any candidate $d \in C \setminus W$.
    By the SDO property, the committee $W$ is Pareto optimal.
\end{proof}

As this algorithm works for any voting instance satisfying SDO, the algorithm also correctly computes Pareto optimal committees satisfying EJR+ given voting instances satisfying Candidate Interval or Voter Interval.

\begin{corollary}
    Given a voting instance $(C,V,\mathcal{A}, k)$ satisfying Candidate Interval, \Cref{alg:greedy-ejr-po} returns a size-$k$ committee that satisfies both EJR+ and Pareto optimality.
\end{corollary}
\begin{proof}
    By \Cref{thm:restricted-candidate-interval-characterization-po}, any voting instance satisfying Candidate Interval also satisfies SDO.
    Thus, by \Cref{thm:sdo-ejrp-po-poly-algo}, \Cref{alg:greedy-ejr-po} correctly returns a size-$k$ committee satisfying both EJR+ and Pareto optimality.
\end{proof}

\begin{corollary}
    Given a voting instance $(C,V,\mathcal{A}, k)$ satisfying Voter Interval, \Cref{alg:greedy-ejr-po} returns a size-$k$ committee that satisfies both EJR+ and Pareto optimality.
\end{corollary}
\begin{proof}
    By \Cref{thm:restricted-candidate-interval-characterization-po}, any voting instance satisfying Candidate Interval also satisfies SDO.
    Thus, by \Cref{thm:sdo-ejrp-po-poly-algo}, \Cref{alg:greedy-ejr-po} correctly returns a size-$k$ committee satisfying both EJR+ and Pareto optimality.
\end{proof}

We thereby resolve the problem of finding a polynomial-time algorithm for finding a committee that satisfies both EJR+ and Pareto optimality for voting instances satisfying Voter Interval.

\paragraph{Reconfiguration Graph}
\label{subsec:sdo-impl-reconfig-graph}

 For SDO voting instances, the Pareto optimality reconfiguration graph is connected, as we now show.
 Further, the distance between two size-$k$ committees $A$ and $B$ is $k - \lvert A \cap B\rvert$.
 This implies any Pareto optimal committee to be reconfigurable into any other without the need for auxiliary candidates, not in either one of the committees.
 Our proof also implies a possible which candidate can be replaced by which candidate in each step, namely the lowest approval score candidate of the original committee by a non-dominated candidate of the target committee.

\begin{lemma}
    Let $(C,V,\mathcal{A}, k)$ be a voting instance for which the SDO property holds.

    Let $W_1, W_2 \subseteq C$ be two size-$k$ Pareto optimal committees.
    The distance between $W_1$ and $W_2$ in the reconfiguration graph $\Gamma_{Po}(\mathcal{A}, k)$ equals $k - |W_1\cap W_2|$.
    \label{lem:restricted-reconfiguration}
\end{lemma}
\begin{proof}
    We show our claim by induction.

    IB: For two committees $W_1, W_2$ that only differ in one candidate, the statement is trivially true as the candidate that is in $W_1$ but not in $W_2$ can be replaced by the candidate that is not in $W_1$ but in $W_2$.

    IH: Let $p \in \N$. Assume that for any size-$k$ Pareto optimal committees $W_1, W_2 \subseteq C$ with $p = k-|W_1 \cap W_2|$, the distance between $W_1$ and $W_2$ is exactly $p$.

    IS:
    Let $W_1$ and $W_2$ be two Pareto optimal committees differing in exactly $p+1$ candidates.
    In particular, no candidate $c \in W_1$ is dominated by any candidate $d \in C \setminus W_1$ and no candidate $c \in W_2$ is dominated by any candidate $d \in C \setminus W_2$.
    We would like to replace in each step one candidate of $W_1$ by a candidate of $W_2$ while maintaining Pareto optimality.

    Let $c\in W_2\setminus W_1$ be a candidate with maximum approval score among all candidates in $W_2\setminus W_1$ and $d \in W_1\setminus W_2$ one with minimum approval score among all candidates in $W_1\setminus W_2$.

    \noindent
    We claim that $W_1-d+c$ is a Pareto optimal committee with $d(W_1-d+c, W_2) = p$.

    As candidate $c$ has a maximum approval score in $W_2\setminus W_1$, $c$ is not dominated by any candidate in $W_2 \setminus W_1$.
    Further, as $W_2$ is Pareto optimal, the candidate $c$ is not also dominated by any candidate in $C\setminus W_2$.
    Then, the candidate $c$ is not dominated by any candidate in $C\setminus W_1$, too, as $W_2 \setminus W_1 \cup C \setminus W_2 = C \setminus (W_1 \cap W_2) \supseteq C \setminus W_1$.

    As $d$ has a minimum approval score in $W_1\setminus W_2$, no candidate in $W_1 \setminus W_2$ is dominated by $d$.
    Additionally, $d \notin W_2$ and $W_2$ is Pareto optimal and therefore no candidate in $W_2$ can be dominated by $d$.
    Thus, candidate $c$ is also not dominated by $d$.

    Further, no candidate in $W_1-d$ is dominated by a candidate in $C \setminus (W_1-d)$, as $W_1$ is Pareto optimal.
    Candidate $c$ also does not dominate any candidate in $W_1-d$.

    Therefore, no candidate in $W_1-d+c$ is dominated by a candidate in $C \setminus (W_1-d+c)$ and hence, by our choice of the voting instance, the committee $W_1-d+c$ is Pareto optimal.
    As $k-\left|(W_1-d+c)\cap W_2\right| = p$ holds, and both committees $W_1-d+c$ and $W_2$ are Pareto optimal, their distance in the Pareto reconfiguration graph equals $p$ by the induction hypothesis.
    Thus, the committees $W_1$ and $W_2$ have a distance of $p+1$ in the reconfiguration graph.
\end{proof}

This can now also be shown to hold for any voting instance satisfying Candidate Interval or Voter Interval.

\begin{corollary}
    Let $(C,V, \mathcal{A}, k)$ be an approval-based voting instance that satisfies the candidate interval property for some total order $\preceq$ on the set of candidates.

    Further, let $W_1, W_2 \subseteq C$ be two size-$k$ Pareto optimal committees.
    The distance between $W_1$ and $W_2$ in the reconfiguration graph $\Gamma_{Po}(\mathcal{A}, k)$ equals $k - \lvert W_1\cap W_2 \rvert$.
    \label{cor:restricted-candidate-interval-reconfiguration}
\end{corollary}
\begin{proof}
    Let $(C, V, \mathcal{A}, k)$ be a voting instances satisfying Candidate Interval.
    By \Cref{thm:restricted-candidate-interval-characterization-po}, the SDO property holds for this voting instance.

    Hence, by \Cref{lem:restricted-reconfiguration}, the distance between two size-$k$ Pareto optimal committees $W_1, W_2 \subseteq C$ in the reconfiguration graph $\Gamma_{Po}(\mathcal{A}, k)$ for this voting instance equals $k - \lvert W_1 \cap W_2 \rvert$.
\end{proof}

\begin{corollary}
    Let $(C,V, \mathcal{A}, k)$ be an approval-based voting instance that satisfies the voter interval property for some total order $\preceq$ on the set of voters.

    Further, let $W_1, W_2 \subseteq C$ be two size-$k$ Pareto optimal committees.
    The distance between $W_1$ and $W_2$ in the reconfiguration graph $\Gamma_{Po}(\mathcal{A}, k)$ equals $k - \lvert W_1\cap W_2 \rvert$.
    \label{cor:restricted-voter-interval-reconfiguration}
\end{corollary}
\begin{proof}
    Let $(C, V, \mathcal{A}, k)$ be a voting instances satisfying Voter Interval.
    By \Cref{thm:restricted-voter-interval-characterization-po}, the SDO property holds for this voting instance.

    Hence, by \Cref{lem:restricted-reconfiguration}, the distance between two size-$k$ Pareto optimal committees $W_1, W_2 \subseteq C$ in the reconfiguration graph $\Gamma_{Po}(\mathcal{A}, k)$ for this voting instance equals $k - \lvert W_1 \cap W_2 \rvert$.
\end{proof}

Thus, the Pareto optimality reconfiguration graph is connected for every voting instance satisfying SDO, such as voting instances satisfying Candidate Interval or Voter interval.

\subsection{Counting Pareto optimal committees in Voter Interval}
\label{chp:counting-voter-interval}
Another question we approached using the SDO property is counting the number of Pareto optimal committees in polynomial time.
In particular, we describe an algorithm using dynamic programming to count the number of Pareto optimal committees for voting instances satisfying Voter Interval, and argue about its correctness.

We believe the number of Pareto optimal can be counted for voting instances satisfying Voter Interval, because in voting instances satisfying Voter Interval every candidate is approved by an interval of voters.
An observation we make is that all candidates dominating one candidate have to be close to each other:
More specifically, we first partition all candidates into levels iteratively.
All candidates that are not dominated by any candidate are assigned to the first level and every other candidate is assigned to one level higher than the candidate that dominates that candidate with the highest level.
Within each level, the candidates dominating the same candidate of a higher level must be consecutive when sorted from left to right by their left-most voter.

This provides a strong structure we can use to count the number of Pareto optimal committees.
Further, we use that voting instances satisfying Voter Interval also satisfy the SDO property.
Thus, we can determine, for each candidate, the set of candidates that have to be included in a Pareto optimal committee such that the candidate can be added to the committee and the committee remains Pareto optimal.
If these sets of candidates are included in the committee for each candidate in the committee, the committee is Pareto optimal.

By our partitioning, described above, we separate candidates into candidates not dominated by any candidates, i.e., the level one candidates, candidates only dominated by level one candidates, i.e., level two candidates, candidates only dominated by candidates in level two and higher, and so on.
We say a committee is Pareto optimal with respect to level $i$ or higher, if all candidates in the committee are in level $i$ or higher, and no other committee, for which all candidates are in level $i$ or higher, es-dominates the committee.

We then compute, for each level $i$, the number of Pareto optimal committees with respect level $i$ or higher.
We partition every Pareto optimal committee $W$ with respect to level $i$ or higher into three disjoint subsets $A, B$ and $C$, such that $B$ is a maximal subset of consecutive level $i$ candidates, when sorted from left to right, including the rightmost candidate in $W$, $C$ is a Pareto optimal committee with respect level $i+1$ or higher, such that every level $i$ candidate, who dominates a candidate in $C$, is in $B$, and $A$ is a Pareto optimal committee with respect to level $i$ candidates or higher.
All Pareto optimal committees can be partitioned uniquely into such sets $A, B$ and $C$.
Further, we can efficiently compute the number of each of these sets and then combine them.
Thus, we can compute the total number of all Pareto optimal committees.

We reformulate this idea into an algorithm using dynamic programming.

\begin{conjecture}
The number of Pareto optimal committees in a voting instance $(C,V, \mathcal{A})$ satisfying voter interval can be counted in polynomial time.
\end{conjecture}

In the following, we assume that no two candidates are approved by the exact same voters, i.e., for no two candidates $a, b \in C$ $V(a) = V(b)$ holds.
We expect this algorithm to also work for voting instances in which multiple candidates are approved by the same voters by introducing new voters such that the domination relation remains unchanged for all candidates.
Voters could for example be added similarly to the way depicted in \Cref{fig:counting-voter-interval-duplicates}.

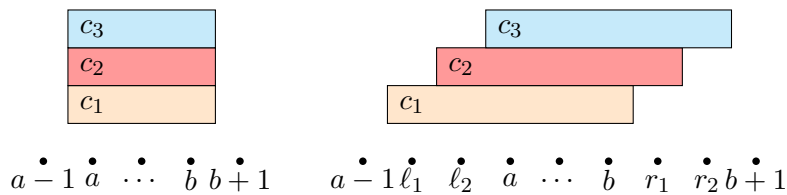
\begin{figure}[h!]
\centering
\begin{tikzpicture}[yscale=0.5, xscale=0.65]
\node[circle, fill=black, inner sep=1pt] at (0.5,-1) {};
\node at (0.5,-1.5) {$a-1$};

\node[circle, fill=black, inner sep=1pt] at (1.5,-1) {};
\node at (1.5,-1.5) {$a$};

\node[circle, fill=black, inner sep=1pt] at (2.5,-1) {};
\node at (2.5,-1.5) {$\dots$};

\node[circle, fill=black, inner sep=1pt] at (3.5,-1) {};
\node at (3.5,-1.5) {$b$};

\node[circle, fill=black, inner sep=1pt] at (4.5,-1) {};
\node at (4.5,-1.5) {$b+1$};

\node[circle, fill=black, inner sep=1pt] at (7,-1) {};
\node at (7,-1.5) {$a-1$};

\node[circle, fill=black, inner sep=1pt] at (8,-1) {};
\node at (8,-1.5) {$\ell_1$};

\node[circle, fill=black, inner sep=1pt] at (9,-1) {};
\node at (9,-1.5) {$\ell_2$};

\node[circle, fill=black, inner sep=1pt] at (10,-1) {};
\node at (10,-1.5) {$a$};

\node[circle, fill=black, inner sep=1pt] at (11,-1) {};
\node at (11,-1.5) {$\dots$};

\node[circle, fill=black, inner sep=1pt] at (12,-1) {};
\node at (12,-1.5) {$b$};

\node[circle, fill=black, inner sep=1pt] at (13,-1) {};
\node at (13,-1.6) {$r_1$};

\node[circle, fill=black, inner sep=1pt] at (14,-1) {};
\node at (14,-1.6) {$r_2$};

\node[circle, fill=black, inner sep=1pt] at (15,-1) {};
\node at (15,-1.5) {$b+1$};

    \drawcand{1}{2/4}
    \drawcand{2}{2/4}
    \drawcand[cyan!20]{3}{2/4}

    \draw[fill=orange!20] (7.5,0) rectangle (12.5,1);
    \node at (8,0.5) {$c_1$};
    \draw[fill=red!40] (8.5,1) rectangle (13.5,2);
    \node at (9,1.5) {$c_2$};
    \draw[fill=cyan!20] (9.5,2) rectangle (14.5,3);
    \node at (10,2.5) {$c_3$};
\end{tikzpicture}
\caption{Voters can likely be added such that the intervals of voters approving each candidate are shifted. This way these candidates do not dominate each other while other domination relations can be preserved.}
\label{fig:counting-voter-interval-duplicates}
\end{figure}

\subsubsection{Partitioning of the Candidates}
We first partition all candidates iteratively into sets as follows:

Let $D_1 \subseteq C$ be the set of all candidates not dominated by any candidate of $C$.
For every other candidate $c \in C \setminus D_1$, let candidate $c$ be in the subset $D_i$, iff a candidate $c' \in D_{i-1}$ dominates candidate $c$ and for every candidate $c'' \in C$ dominating candidate $c$ an index $j$ exists such that $j < i$ and $c \in D_j$, i.e., all candidates dominating $c$ are in a lower level and at least one is exactly one level below $c$.

An example of this partitioning can be seen in \Cref{fig:counting-voter-interval-levels}.

\begin{figure}
    \centering

\begin{tikzpicture}[yscale=0.5, xscale=0.65, voter/.style={anchor=south}]

    \foreach \i in {1,...,10}
        \node[voter] at (\i-0.5, -1.5) {$\i$};

    \drawcand{1}{8/8}

    \drawcand{2}{2/4}
    \drawcand[violet!30]{3}{4/5}
    \drawcand[blue!20]{4}{7/9}

    \drawcand[green!20]{5}{1/7}
    \drawcand[lime!20]{6}{5/10}

    \end{tikzpicture}

    \caption{The candidates $c_5$ and $c_6$ make up $D_1$ as they are not dominated by any candidate, candidates $c_2, c_3$ and $c_4$ are in $D_2$ because they are only dominated by candidates in $D_1$ and candidate $c_1$ is in $D_3$ here as it is dominated by candidates $c_4$ and $c_6$.}
    \label{fig:counting-voter-interval-levels}
\end{figure}
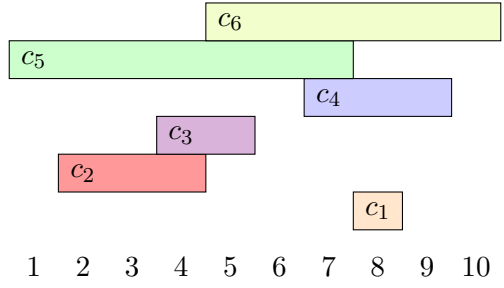

We define $level[c]$ for each candidate $c \in C$ as the index of the level $c$ is in, i.e., $level[c] \coloneqq  i$, iff $c \in D_i$ holds.
The number of candidates $m$ is an upper bound on the number of levels.
Further, we observe that every candidate is dominated by at least one candidate of each lower level due to Pareto dominating being a transitive relation.

\subsubsection{Definition of the Tables}

For our algorithm we use two tables which we define in the following, starting with the $dp2$ table:

\paragraph{Table $dp2$}
For any voter $v_\alpha, v_\beta, v_f, v_g \in V$, size $s \in \mathbb{N}$ and level $t \in \mathbb{N}$, the value \linebreak[4]$dp2[v_\alpha][v_\beta][v_f][v_g][s][t]$ is the number of size-$s$ committees $W$ for which the following holds:
\begin{enumerate}
    \item  No candidate of a lower level than $t$ can be contained in the committee, i.e., $\forall c \in W \colon level[c] \geq t$.
    \item $W$ is Pareto optimal with respect to $\bigcup_{i \geq t} D_i$.
    Because of the SDO property, this is equivalent to no candidate $d \in C \setminus W$ with $level[d] \geq t$ dominating any candidate $c \in W$.
    \item A candidate $d \in D_t$ is in the committee $W$, iff all voters approving candidate $d$ lie between $v_\alpha$ and $v_\beta$ inclusively, i.e.,

    \noindent
    $\forall v \in V(d)\colon v_\alpha \preceq v \preceq v_\beta$.
    \item A candidate of level $t$ whose leftmost voter is $v_\alpha$ is in the committee $W$, i.e., $\exists c \in W\colon v_\alpha \in V(c) \land level[c] = t$.
    \item A candidate of level $t$ whose rightmost voter is $v_\beta$ is in the committee $W$, i.e., $\exists c \in W\colon  v_\beta \in V(c) \land level[c] = t$.
    \item For every candidate $c \in W$ in the committee, the rightmost voter approving $c$ is right of $v_f$ and the leftmost left of $v_g$, i.e.,

    \noindent
    $\forall c \in W\colon \exists v_1 \in V(c)\colon v_f \prec v_1 \land \exists v_2 \in V(c) \colon v_2 \prec v_g$.
\end{enumerate}

Less formally speaking, $dp2[v_\alpha][v_\beta][v_f][v_g][s][t]$ describes the number of Pareto optimal committees of size $s$, disregarding the candidates in lower levels than $t$, in which all level $t$ candidates between $v_\alpha$ and $v_\beta$ are included but no candidate only approved by voters left of $v_{f+1}$ or only by voters right of $v_{g-1}$.

Additionally, we can make the following observations:
As each candidate $c$ of a level $i > t$ is dominated by a candidate of level $t$ and for every level-$t$ candidate in the committee $d \in W \cap D_t$ all voters approving candidate $d$ lie between voters $v_\alpha$ and $v_\beta$, all voters approving candidate $c$ also lie between $v_\alpha$ and $v_\beta$, i.e., \linebreak[4]$\forall v \in V(c)\colon v_\alpha \preceq v \preceq v_\beta$ holds.
However, there can be a candidate $c \in D_i$ for $i > t$ whose voters all lie between voters $v_\alpha$ and $v_\beta$ but who is dominated by a candidate $d \in D_t$, approved by a voter outside of $v_\alpha$ and $v_\beta$.
Because candidate $d$ cannot be included in a committee counted because of the third property, candidate $c$ may also not be included in any committee counted because, by the second property, despite only being approved by voters between $v_\alpha$ and $v_\beta$.

\paragraph{Table $dp1$}
Further, let $dp1[v_\alpha][v_\beta][v_f][v_g][s][t]$ be defined as the number of size-$s$ committees $W$ for which the following holds:
\begin{enumerate}
    \item No candidate of a lower level than $t$ can be contained in the committee, i.e., $\forall c \in W \colon level[c] \geq t$.
    \item $W$ is Pareto optimal with respect to $\bigcup_{i \geq t} D_i$.
    As for $dp2$, because of the SDO property, this is equivalent to no candidate $d \in C \setminus W$ with $level[d] \geq t$ dominating any candidate $c \in W$.
    \item No candidate in $W$ is approved by a voter outside of $v_\alpha$ to $v_\beta$, i.e.,

    \noindent
    $\forall c \in W\colon \forall v \in V(c)\colon v_\alpha \preceq v \preceq v_\beta$.
    \item A candidate of level $t$ whose rightmost voter is $v_\beta$ is in the committee $W$, i.e., $\exists c \in W\colon  v_\beta \in V(c) \land level[c] = t$.
    \item For every candidate $c \in W$ in the committee, the rightmost voter approving $c$ is right of $v_f$ and the leftmost left of $v_g$, i.e.,

    \noindent
    $\forall c \in W\colon \exists v_1 \in V(c)\colon v_f \prec v_1 \land \exists v_2 \in V(c) \colon v_2 \prec v_g$.
\end{enumerate}

In contrast to $dp2$, not all candidates of level $t$ only approved by voters between $v_\alpha$ and $v_\beta$ have to be included in committees counted by $dp1$.
By the second property, candidates in a committee counted by $dp1$ must still only be approved by voters between $v_\alpha$ and $v_\beta$.
However, only a level-$t$ candidate approved by voter $v_\beta$ must be included in any committee counted by $dp1$.

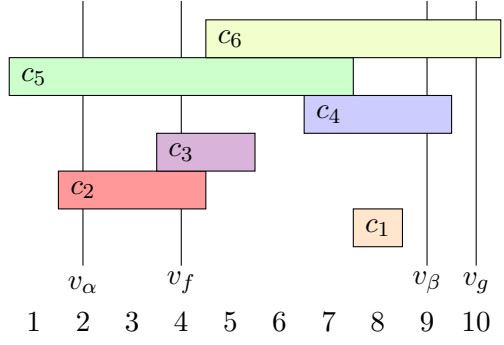
\begin{figure}[h]
\centering
\begin{tikzpicture}[yscale=0.5, xscale=0.65, voter/.style={anchor=south}]

    \foreach \i in {1,...,10}
        \node[voter] at (\i-0.5, -2.5) {$\i$};

    \drawcand{1}{8/8}

    \drawcand{2}{2/4}
    \drawcand[violet!30]{3}{4/5}
    \drawcand[blue!20]{4}{7/9}

    \drawcand[green!20]{5}{1/7}
    \drawcand[lime!20]{6}{5/10}

    \begin{scope}[on background layer]
    \draw (1.5,-0.5) -- (1.5,6.5);
    \node[voter] at (1.5,-1.5) {$v_\alpha$};

    \draw (3.5,-0.5) -- (3.5,6.5);
    \node[voter] at (3.5,-1.5) {$v_f$};

    \draw (8.5,-0.5) -- (8.5,6.5);
    \node[voter] at (8.5,-1.5) {$v_\beta$};

    \draw (9.5,-0.5) -- (9.5,6.5);
    \node[voter] at (9.5,-1.5) {$v_g$};
    \end{scope}

    \end{tikzpicture}

    \caption{For every committee counted by the value of  $dp1[v_\alpha][v_\beta][v_f][v_g][s][2]$, candidate $c_4$ must be included, as $c_4$ is in level $2$ and also approved by voter $v_\beta$. Candidates $c_1$ and $c_3$ may be included, as no condition is violated for both of them. Candidate $c_2$ cannot be included in any committee counted by this value, as no voter right of $v_f$ approves $c_2$.
    For $dp2[v_\alpha][v_\beta][v_f][v_g][s][2]$, no committee exists that satisfies all conditions as candidate $c_2$ is not approved by any voter right of $v_f$ but is in level $t$ and only approved by voters between $v_\alpha$ and $v_\beta$.}
    \label{fig:counting-voter-interval-limits}
\end{figure}

\subsubsection{Computation of $dp1$}
\paragraph{Initialization}
There is always exactly one Pareto optimal committee of size 0 given any other parameters, as the empty set is a Pareto optimal size-0 committee.
Thus, we set all values of the array for sizes 0, i.e., where $s = 0$, as 1.

Further, we set all values for level $t > m$ to 0 if $s > 0$, since no candidates can have a level of more than $m$.

Next, we set the value of $dp1[v_\alpha][v_\beta][v_f][v_g][s][t]$ to $0$, if $v_\beta \prec v_f$ holds, because the rightmost candidate, whose voters are all between $v_\alpha$ and $v_\beta$, has to be included in every committee counted but all voters approving that candidate are left of $v_f$ as $v_\beta \prec v_f$ holds.
Hence, no such committee exists.
Analogously, we set the value of $dp1[v_\alpha][v_\beta][v_f][v_g][s][t]$ to $0$, if $v_g \prec v_\alpha$ holds.

This concludes the initialization of $dp1$.

\paragraph{Computation}

In all other cases, the value of $dp1[v_\alpha][v_\beta][v_f][v_g][s][t]$ can be computed based on the values of $dp2$ for level $t$ as follows:

Let $c^t_1, \dots, c^t_{x_t} \in C$ be all candidates of level $t$ such that $\forall v \in V(c^t_i):v_\alpha \preceq v \preceq v_\beta$ holds sorted from left to right by their leftmost voter.
I.e., candidate $c^t_i$ is approved by a voter left of all voters approving $c^t_j$ for every $j > i$.
As no candidates are approved by only the same voters and no two candidates within one level dominate each other, this sorting is unique.

The candidates $c^t_1, \dots, c^t_{x_t}$ are the level $t$ candidates that may be included in a committee we want to count.

If $v_\beta \notin V(c^t_{x_t})$ is true, then we set the value of $dp1[v_\alpha][v_\beta][v_f][v_g][s][t]$ to $0$.
Else, $c^t_{x_t}$ must be included in every committee counted.
Let $\ell^t_i$ be the leftmost voter approving candidate $c^t_i$ and $r^t_i$ be the rightmost voter approving $c^t_i$.
Hence, for example, $r^t_{x_t} = v_\beta$.
We calculate the values of $dp1$ as following:
\begin{align*}
    dp1[v_\alpha][v_\beta][v_f][v_g][s][t] \coloneqq &\phantom{+ } \sum_{i=1}^{x_t}dp2[\ell^t_i][v_\beta][v_f][v_g][s][t] \\*
    &+ \sum_{q=1}^{s-1} \sum_{i = 1}^{x_t} \sum_{j = 1}^{i-2} dp1[v_\alpha][r^t_j][v_f][v_g][q][t] \\*
    &\phantom{+ \sum_{q=1}^{s-1} \sum_{i=1}^{x_t} \sum_{j=1}^{i-2}}
   \cdot dp2[\ell^t_i][v_\beta][v_f][v_g][s-q][t]
\end{align*}

\paragraph{Explanation}
We now argue for why the value of $dp1[v_\alpha][v_\beta][v_f][v_g][s][t]$ is computed correctly by the equation above.

In a committee $W$ that is supposed to be counted by $dp[v_\alpha][v_\beta][v_f][v_g][s][t]$ some, but possibly not all candidates of level $t$ whose voters are between $v_\alpha$ and $v_\beta$, i.e., the candidates $c^t_1,\dots, c^t_{x_t}$, are included.
However, candidate $c^t_{x_t}$ has to be in the committee $W$.
The candidates of level $t$ in $W$ can be partitioned into two sets.
First a maximum set of consecutive candidates $A$ including $c^t_{x_t}$, e.g., $\{c^t_{x_t-2}, c^t_{x_t-1}, c^t_{x_t}\}$, and second a set $B$, the rest of the level $t$ candidates in the committee $W$.
For an example, consider a set of candidates $\{c_1, \dots, c_8\}$ as depicted in \Cref{fig:counting-voter-interval-dp1-partitioning} and a committee $W$ that contains candidates $c_1, c_2, c_4, c_6, c_7$ and $c_8$ but not candidates $c_3$ or $c_5$.
In that case, set $A$ contains the candidates $c_6, c_7$ and $c_8$, as candidate $c_5$ is not in the committee and, thus, this is the largest set of consecutive candidates that contains candidate $c_8$, the rightmost candidate.
Hence, set $B$ contains the candidates $c_1, c_2$ and $c_4$.

\begin{figure}[h]
    \centering
    \begin{tikzpicture}[yscale=0.5, xscale=0.65, voter/.style={anchor=south}]

    \foreach \i in {1,...,18}
        \node[voter] at (\i-0.5, -2.5) {$\i$};
        \drawcand{1}{1/3}
        \drawcand{2}{4/7}
        \drawcand[magenta!20]{3}{5/8}
        \drawcand[violet!30]{4}{7/9}
        \drawcand[blue!20]{5}{11/12}
        \drawcand[cyan!20]{6}{12/14}
        \drawcand[green!20]{7}{13/16}
        \drawcand[yellow!30]{8}{15/18}

    \end{tikzpicture}
    \caption{In this figure, only the candidates $c^t_1, \dots, c^t_{x_t}$ for some choice of level $t$ and voters $v_\alpha$ and $v_\beta$ are depicted.}
    \label{fig:counting-voter-interval-dp1-partitioning}
\end{figure}
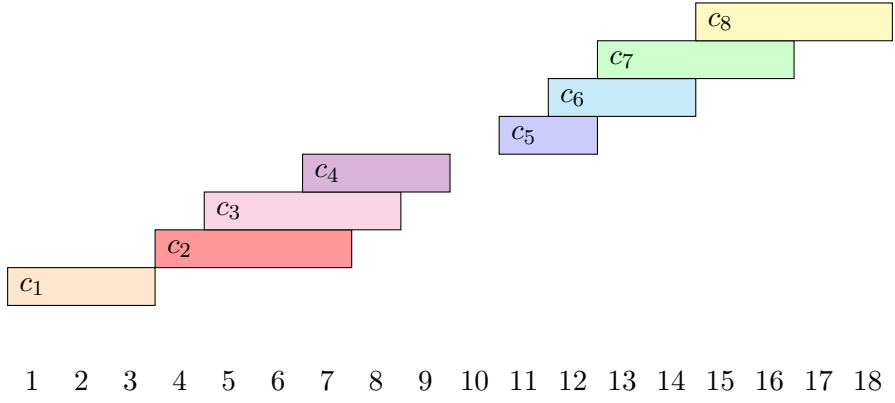

As $c^t_{x_t} \in W$ holds, the set $A$ cannot be empty.
The set $B$, however, can be empty, meaning that all candidates of level $t$ are consecutive in the committee $W$.

These committees, where $B =\emptyset$ is true, are counted by the first sum.
We iterate over all possible leftmost candidates of the consecutive level $t$ candidates in a committee and sum up the values of $dp2$ corresponding to the number of committees with the same set of level $t$ candidates with the same exclusion limits $v_f$ and $v_g$.
For each possible leftmost candidate $c^t_i$ included in a committee, we use the leftmost voter approving $c^t_i$, i.e., $\ell_i$, as the voter for $v_\alpha$ for $dp2$.
By the definition of $dp2$, the committees counted for $\ell_i$ as $v_\alpha$, the same voters $v_\beta, v_f$ and $v_g$, the same size $s$ and the same level $t$, are committees in which all level $t$ candidates are consecutive and include the candidates $c^t_i$ and $c^t_{x_t}$.

In case the set $B$ is not empty, there is a candidate $c^t_d$ which is not in $W$ but $c^t_{d+1}$ is in $W$ and in the first set.
The candidate $c^t_d$ marks the end of the consecutive candidates of the first set $A$.
In our example above, candidate $c_5$ is this candidate $c^t_d$, as $c_5$ is not in the committee but $c_6$ is in the committee and set $A$.
This candidate, so to say, forms a gap between the candidates of set $A$ and set $B$.

Further, another candidate $c^t_e$ exists that is the rightmost candidate in the second set $B$.
In our previous example, candidate $c_4$ is this candidate $c^t_e$ as no candidate right of $c_4$ is in set $B$.

Every candidate $c_i \in W$ with level $t' > t$ can be sorted into two sets depending on whether it is dominated by level t candidates of set $A$ or set $B$.
No candidate is dominated by a candidate in set $A$ and one in set $B$.

\begin{lemma}
    Let $W \subseteq C$ be a Pareto optimal committee and $t \in \N$.
    Further, let $c \in W$ be some candidate with $level[c] > t$.
    Let the candidates in $D_t$ be sorted left to right by their leftmost voter.
    Then, the candidates in $D_t$ that dominate $c$ form an interval, i.e., there are no three candidates $d_1, d_2, d_3 \in D_t$ with candidate $d_1$ being left of $d_2$ and candidate $d_2$ left of $d_3$, such that candidates $d_1$ and $d_3$ dominate $c$ but candidate $d_2$ does not.
    \label{lem:counting-voter-interval-partition}
\end{lemma}
\begin{proof}
    For the sake of contradiction, assume there are such candidates $d_1, d_2, d_3 \in D_t$.
    Let $v_{\ell_1}$ be the leftmost voter of $d_1$ and $v_{r_1}$ the rightmost one.
    Analogously, let $v_{\ell_2}, v_{\ell_3}$ be the leftmost voters of candidates $d_2$ and $d_3$, and $v_{r_2}, v_{r_3}$ be the rightmost voters of candidates $d_2$ and $d_3$.
    As candidate $d_1$ is left of $d_2$, $v_{\ell_1} \prec v_{\ell_2}$ holds for their leftmost voters, and so does $v_{\ell_2} \prec v_{\ell_3}$ as candidate $d_2$ is left of $d_3$.
    Because candidate $c$ is dominated by $d_3$, $v_{\ell_3} \preceq v$ holds for every voter $v \in V(c)$.
    Hence, $v_{\ell_2} \prec v$ holds for every voter $v \in V(c)$.

    Candidate $d_2$ does not dominate $c$ by our assumption.
    Thus, $v_{r_2} \prec v_r$ has to hold for the rightmost voter $v_r$ approving $c$.
    Further, candidate $d_1$ dominates $c$ and hence, $v_r \preceq v_{r_1}$ holds.
    By $v_{l_1} \prec v_{l_2}$ and $v_{r_2} \prec v_{r_1}$ we can conclude that candidate $d_2$ is dominated by candidate $d_1$ as by the voter interval property candidate $d_1$ is approved by every candidate approving $d_2$ and another voter, e.g. voter $v_{l_1}$.
    This contradicts the choice of $d_1$ and $d_2$ as candidates of the same level, as candidates of the same level cannot dominate each other by the definition of the levels.

    Hence, no such three candidates can exist and the claim follows.
\end{proof}

Applying \Cref{lem:counting-voter-interval-partition}, we can extend sets $A$ and $B$ as described above.
Every candidate $c \in W$ with $level[c] > t$, for which a candidate $d \in A$ dominating candidate $c$ exists, is put into set $A$.
If no such candidate exists, a candidate \linebreak[4]$d \in B$ dominates $c$ and candidate $c$ is put into set $B$.
In the following, we call the candidates in $A$ for a committee $W$ the \emph{right candidates} and the candidates in $B$ the \emph{left candidates}.

We count these committees, for which set $B$ is not empty, in the second sum.
In the outer most sum, we iterate over the number of candidates in $B$, denoted as $q$.
In the middle sum, we iterate over the leftmost candidate of level $t$ included in the set $A$, similarly to the first sum described above.
This marks the left end of the right candidates of level $t$.
In the inner sum, we iterate over the rightmost candidates of level $t$ in the set $B$, i.e., the right end of the left candidates.
Thus, we iterate over every possible gap of consecutive level $t$ candidates between the left and right candidates.
By the value $dp1[v_\alpha][r^t_j][v_f][v_g][q][t]$ in the sum, we count every size-$q$ set $B$ of a committee.
For every committee $W$ that is to be counted, the candidates of set $B$ are approved by at least one voter between voters $v_f$ and $v_g$ and all voters approving level $t$ candidates in $B$ have to be right of $v_\alpha$ because of the definition of $dp1$.
Further, no candidate $c \in B$ is dominated by a candidate $d \in \bigcup_{i \geq t}D_i$.
Thus, the set $B$ itself is a size-$q$ Pareto optimal committee as the SDO property holds.
As no further restrictions have been made to the set $B$ besides the rightmost level $t$ candidate $c^t_j$ included in $B$, it is a committee that is counted by the value $dp1[v_\alpha][r^t_j][v_f][v_g][q][t]$.
Vice versa, each of the committees counted by this value is a subset $B$ of some committee $W$ that is supposed to be counted by $dp1[v_\alpha][v_\beta][v_f][v_g][s][t]$.

As in the previous case of $B$ being empty, the number of sets $A$ of size $s-q$ is counted by the value $dp2[\ell^t_i][v_\beta][v_f][v_g][s-q][t]$ because $A$ is a size-$(s-q)$ Pareto optimal committee with $c^t_i$ as the leftmost level $t$ candidate included and only consecutive level $t$ candidates following $c^t_i$.

Given the size $q$ and the two candidates the sums iterate over, any combination of valid sets $B$ and sets $A$ forms a committee $W$ that is supposed to be counted by $dp1[v_\alpha][v_\beta][v_f][v_g][s][t]$.
As the sets $A$ and $B$ are Pareto optimal committees with respect to the candidates of level $t$ and higher $\bigcup_{i \geq t} D_i$ and the SDO property holds, no candidate in $A$ or $B$ is dominated by any candidate $d \in \bigcup_{i \geq t} D_i$.
Hence, no candidate $c \in W$ is dominated by any candidate $d \in \bigcup_{i \geq t}D_t$ for any combination $W$ of two such sets $A$ and $B$.
Further, the committee $W$ is Pareto optimal with respect to $\bigcup_{i \geq t}D_i$ because the SDO property holds.

By iterating over all possible numbers of left candidates and the possible gap between the rightmost level $t$ candidate on the left and the leftmost level $t$ candidate on the right, every committee that is supposed to be counted by $dp1[v_\alpha][v_\beta][v_f][v_g][s][t]$ is counted.

Thus, exactly the committees that are supposed to be counted are counted.

\subsubsection{Computation of $dp2$}

\paragraph{Initialization}
Analogously to $dp1$, we set all values of the array for sizes 0 to 1 and all values for level $t > m$ to 0.
We set the value of $dp2[v_\alpha][v_\beta][v_f][v_g][s][t]$ to $0$, if $v_\beta \prec v_f$ or $v_g \prec v_\alpha$ holds or $v_\beta \prec v_\alpha$.

Additionally, we set the value of $dp2[v_\alpha][v_\beta][v_f][v_g][s][t]$ to 0, if \linebreak[4]$\forall d \in D_t\colon \exists v \in V(d)\colon v \prec v_\alpha \lor \forall d \in D_t\colon \exists v \in V(d)\colon v_\beta \prec v$ holds, i.e.\ if $v_\alpha$ is not the leftmost voter of a candidate of level $t$ or $v_\beta$ is not the rightmost.

This concludes the initialization of $dp2$.

\paragraph{Computation}

In all other cases, the value of $dp2[v_\alpha][v_\beta][v_f][v_g][s][t]$ can be computed based on the values of $dp1$ for level $t+1$ as follows:

Similarly to the computation of $dp1$, let $c^{u}_1, \dots c^{u}_{x_{u}}$ be all candidates of level $u$, for which all voters approving each candidate lie between voters $v_\alpha$ and $v_\beta$, sorted from left to right by their leftmost voter.
Let $\ell^u_{i}$ be the leftmost approving candidate $c^u_i$ and $r^u_i$ the rightmost one.
Let $c^t_{left}$ be the candidate left of $c^t_1$ and $r^t_{left}$ be the rightmost voter approving candidate $c^t_{left}$ and $c^t_{right}$ be the candidate right of $c^t_{x_t}$ and $\ell^t_{right}$ the leftmost voter approving candidate $c^t_{right}$.

Further, let $q$ be the number of level $t$ candidates $d \in D_t$ with \linebreak[4] $\forall v \in V(d): v_\alpha \preceq v \preceq v_\beta$, i.e., the level $t$ candidates that have to be included in a committee counted by $dp2[v_\alpha][v_\beta][v_f][v_g][s][t]$.

We make a case distinction by the number of these candidates $q$ compared to the numbers in the committees to be counted $s$.

If $s = q$ holds, then $dp2[v_\alpha][v_\beta][v_f][v_g][s][t] \coloneqq 1$ as the only valid committee is $\{c^t_1, \dots, c^t_{x_t}\}$ in this case.
Non of these candidates is dominated by a candidate $d \in \bigcup_{i \geq t}d_i$ by the definition of the levels and the committee thus Pareto optimal with respect to the candidates in $\bigcup_{i \geq t}d_i$, as the SDO property holds.

If $s < q$ holds, then $dp2[v_\alpha][v_\beta][v_f][v_g][s][t] \coloneqq 0$ as in this case not all candidates $\{c^t_1, \dots, c^t_{x_t}\}$ can be included in the committee, violating the properties of committees counted by the value $dp2[v_\alpha][v_\beta][v_f][v_g][s][t]$.

If $s > q$ holds, then we can calculate the value of $dp2[v_\alpha][v_\beta][v_f][v_g][s][t]$ by
\begin{align*}
&dp2[v_\alpha][v_\beta][v_f][v_g][s][t] \\
&\coloneqq \sum_{i=1}^{x_{t+1}} dp1[v_\alpha][r^{t+1}_i][\max(r^t_{left}, v_f)][\min(\ell^t_{right}, v_g)][s-q][t+1].
\end{align*}

\paragraph{Explanation}
In the following, we argue why the value of
$dp2[v_\alpha][v_\beta][v_f][v_g][s][t]$ is computed correctly by the sum above, if $s > q$ holds.

As the $q$ candidates of level $t$, whose voters approving them are all between $v_\alpha$ and $v_\beta$, have to be included in any committee counted, $s-q$ many candidates $c$ with $level[c] > t$ have to be included in every committee that is supposed to be counted by $dp2[v_\alpha][v_\beta][v_f][v_g][s][t]$.
For each candidate $c$ with $level[c] > t$ that is part of a committee $W$ that is supposed to be counted, no candidate $d \in C \setminus W$ with $level[d] \geq t$ may dominate $c$.
By the SDO property, exactly the committees for which the above holds are Pareto optimal with respect to $\bigcup_{i \geq t}D_i$, as required by the definition of $dp2$.

Each of these remaining committees thus $W'$ and the level $t$ candidates make up a unique size-$s$ committee $W = W' \cup \{c^t_1, \dots c^t_{x_t}\}$ that is supposed to be counted by the value of $dp2[v_\alpha][v_\beta][v_f][v_g][s][t]$.
We can calculate the number of such committees by the sum above.

For any committee supposed to be counted by the value $dp2[v_\alpha][v_\beta][v_f][v_g][s][t]$, no candidate may be dominated by a candidate $d \in \bigcup_{i \geq t}D_i$ that is not in the committee.
By the definition of $dp1$, all committees counted in the sum only contain candidates $c$ with $level[c] \geq t+1$ and are Pareto optimal with respect to $\bigcup_{i \geq t+1}D_i$.
Further, all candidates $c^t_1, \dots, c^t_{x_t}$ are in the committee supposed to be counted by corresponding value of $dp2$, by the definition of $dp2$.
Thus, it remains to show that no candidate in a committee counted by the value of $dp1[v_\alpha][r^{t+1}_i][\max(r^t_{left}, v_f)][\min(\ell^t_{right}, v_g)][s-q][t+1]$ for some candidate index $1 \leq i \leq x_{t+1}$ is dominated by any candidate $d \in D_t$ who is approved by a voter left of $v_\alpha$ or right of $v_\beta$.

We make sure this holds for all committees counted by choosing $\max(r^t_{left}, v_f)$ and $\min(\ell^t_{right}, v_g)$ as new limits $v_f$ and $v_g$ for $dp1$ in order to prevent including level $t+1$ candidates who are dominated by such level $t$ candidates, and show that vice versa no candidate violating these limits can be included in a committee supposed to be counted by the value of $dp2[v_\alpha][v_\beta][v_f][v_g][s][t]$.

\begin{lemma}
    Let $t \in \mathbb{N}$ and $D_t = \{c_1, \dots, c_x\}$ be the set of all candidates of level $t$ sorted from left to right by their leftmost voters.
    Further, let $r_i \in V$ be the rightmost voter approving candidate $c_i$ and $\ell_i \in V$ the leftmost one.
    Let $c \in C$ be some candidate such that $level[c] > t$ holds.
    Let $i \in \N$ be maximum such that $v \preceq r_i$ holds for all voters $v \in V(c)$, i.e., $c_i$ is the rightmost candidate of level $t$ for which no voter approving $c$ is right of all voters approving $c_i$.
    Then, candidate $c$ is dominated by some candidate $c_j$ with $j \leq i$.

    Analogously, let $i' \in \N$ be minimum such that $\ell_{i'} \preceq v$ holds for all voters $v \in V(c)$.
    Then, candidate $c$ is dominated by some candidate $c_j$ with $j \geq i'$.
    \label{lem:counting-voter-interval-limits}
\end{lemma}
\begin{proof}
    Let $c \in C$ be some candidate such that $level[c] > t$ holds, and let $i \in \N$ be maximum such that $v \preceq r_i$ holds for all voters $v \in V(c)$.
    If for every voter $v \in V(c)$, $v \preceq \ell_i$ also holds, then, as the two candidates are in different levels, candidate $c_i$ dominates $c$ and therefore the statement follows.

    For the sake of contradiction, assume some voter $v' \in V(c)$ with $v' \prec \ell_i$ exists.
    By the definition of the levels and the transitivity of Pareto optimality, a candidate $c_h \in D_t$ who dominates candidate $c$ exists.

    Further, there is either a candidate $c_g \in D_t$ with $g > i$ that dominates candidate $c$, or for every candidate $c_j \in D_t$ dominating candidate $c$ it holds that $j \leq i$.
    In the second case, some candidate $c_j \in D_t$ with $j \leq i$ dominates candidate $c$.
    It remains to show that such a candidate exists in the first case as well.

    Let $c_g \in D_t$ with $g > i$ be a candidate dominating candidate $c$.
    Hence, for every voter $v \in V(c)$, by the Voter Interval property, $\ell_g \preceq v \preceq r_g$ holds.
    In particular, $\ell_g \preceq v'$ holds as well.

    As $i < g$ holds and the candidates $c_1, \dots, c_x$ are ordered from left to right by their leftmost voter, which is equivalent to ordering them by their rightmost voter because no two candidates dominate each other, $r_i \prec r_g$ holds.

    \noindent
    Thus, $\ell_g \preceq v' \prec \ell_i \preceq r_i \preceq r_g$ holds.
    By the Voter Interval property, we can conclude that every voter who approves candidate $c_i$ also approves $c_g$, and that candidate $c_g$ is approved by another voter as well.
    Hence, candidate $c_i$ is dominated by $c_g$ contradicting $level[c_i] = level[c_g] = t$ as candidates of the same level cannot dominate each other.

    Thus, a candidate $c_j$, with $j \leq i$, dominates candidate $c$.

    The second statement follows analogously.
\end{proof}

Applying \Cref{lem:counting-voter-interval-limits}, we can conclude that all candidates $c \in \bigcup_{i \geq t+1}D_i$, for which $\forall v \in V(c)\colon v \preceq r_{left}$ or $\forall v \in V(c) \colon \ell_{right} \preceq v$ holds, are dominated by a level $t$ candidate whose voters are not all between $v_\alpha$ and $v_\beta$ and therefore, they are correctly excluded from the remaining committees.
Any candidate $c \in \bigcup_{i \geq t+1}D_i$ approved by voters $v_1$ and $v_2$ which may be the same voter, with $r_{left} \prec v_1$ and $\ell_{right} \prec v_2$ also is not dominated by any level $t$ candidate $d \in D_t$ who is not only approved by voters between voter $v_\alpha$ and $v_\beta$ as candidate $c$ is approved by voters $v_1$ and $v_2$ and, by the choice of $r_{left}$ and $\ell_{right}$, candidate $d$ is not approved by voter $v_1$ nor $v_2$.

Further, no remaining committee is counted twice as we iterate over all level $t+1$ candidates as the rightmost candidate included in the remaining committee making the committees counted in each summand unique.

Hence, all committees ought to be counted by $dp2[v_\alpha][v_\beta][v_f][v_g][s][t]$ are counted exactly once.

\subsubsection{Order of Computation}

We start computing the values of both tables at level $m$ and continue in decreasing order.
For each level $t$, we first calculate the values of $dp2$.
The values of $dp2$ only rely on the values of $dp1$ of the previous level $t+1$.
The order of computation given a level $t$ does not matter when computing the values of $dp2$ and can be chosen arbitrarily.

We then compute the values of $dp1$ by iterating over the voters from left to right for $v_\alpha$ and given a voter $v_\alpha$ iterating over all voters right of voter $v_\alpha$ from left to right for $v_\beta$.
Given a pair of voters $v_\alpha, v_\beta$, we iterate over the size of the committees $s$ from $1$ to $k$.

Although this explanation does not fully prove the correctness of the algorithm, we expect the algorithm to be correct.
We have demonstrated how Pareto optimal committees can be constructed using the partitioning of all candidates into levels as defined above.
This can be used when trying to construct Pareto optimal committees satisfying further properties.

\section{Unrestricted Domain}
\label{prt:general-case}

\subsection{Committee Monotonicity}
\label{chp:general-monotonicity}

In this chapter, we discuss the additional challenges when studying committee monotonicity for Pareto optimal committees in the unrestricted domain.
In the restricted domains considered earlier, we showed that adding a non-dominated candidate to a Pareto optimal committee always results in another Pareto optimal committee, without the need to reconsider the original candidates in the committee.
However, in general, there are voting instances $(C, V, \mathcal{A}, k)$ for which a candidate $c \in C$ and a size-$k$ Pareto optimal committee $W$ exist such that candidate $c$ is not dominated by any candidate in $C \setminus W$, yet $W+c$ is not Pareto optimal .
For an example, see the approval profile given by \Cref{fig:general-monotonicity-sdo-counterexample} for a voting instance.

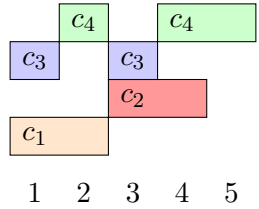
\begin{figure}[h]
\centering
\begin{tikzpicture}[yscale=0.5, xscale=0.65, voter/.style={anchor=south}]

    \foreach \i in {1,...,5}
        \node[voter] at (\i-0.5, -1.5) {$\i$};

    \drawcand{1}{1/2}
    \drawcand{2}{3/4}
    \drawcand[blue!20]{3}{1/1,3/3}
    \drawcand[green!20]{4}{2/2,4/5}
\end{tikzpicture}
\caption{For this depicted approval profile, the committee $W = \{c_1\}$ is Pareto optimal, candidate $c_2$ is not dominated, yet $W+c_2$ is not a Pareto optimal committee, as it is es-dominated by $\{c_3, c_4\}$.}
\label{fig:general-monotonicity-sdo-counterexample}
\end{figure}

We expect that committee monotonicity should still hold for Pareto optimal committees in the unrestricted domain.
If committee monotonicity does not hold, then this would imply that a Pareto optimal committee exists, such that, for each candidate not in the committee, a subset of the committee together with that candidate is es-dominated by a subset of the remaining candidates.

To demonstrate why we believe this is a contradiction, we construct a directed graph with the candidates as node and an edge from candidate $a$ to candidate $b$ if, candidate $a$ is in an es-dominating set if candidate $b$ is added to the committee.
If there is a sink in this graph, then this candidate can be added to the committee while preserving Pareto optimality.
Else, every path in this graph eventually enters a cycle.
This means that, for each candidate not in the committee, a subset of the remaining candidates $A$ exists, such that $A$ dominates a subset $B \subseteq W$ with $\lvert B\rvert = \lvert A \rvert - 1$ and the candidate, however for each candidate in $A$ another such subset exists only containing other remaining candidates.
This appears contradictory because, if a candidate $c$ is in a such a subset of the remaining candidates that dominates a subset of the committee and another candidate, there must be at least one voter $v$ approving this candidate and a candidate in the dominated subset.
Thus, a candidate in a subset of the remaining candidates that dominates a subset of the committee and candidate $c$ must be approved by voter $v$ and another voter $v'$ approving a candidate in the dominated subset.
Next, for this candidate, in a subset of remaining candidates dominating this candidate and a subset of the committee, must be a candidate approved by voter $v$ and a voter $v''$ approving a candidate in the dominated subset and a candidate approved by voter $v'$ and another voter approving a candidate in the dominated subset of the committee.
This way, the number of candidates each voter approves among the remaining candidates appears to have to be at least as high as the number of candidates approved in a subset of the committee, suggesting it was not Pareto optimal to begin with.

\subsubsection{Few relevant Candidates}

A central difficulty in verifying Pareto optimality is the quickly increasing number of subsets of a committee that have to be considered.
For a committee to be Pareto optimal, no subset may be es-dominated by a subset of the remaining candidates.
This requires checking, for each subset of the committee $A$ and subset of the remaining candidates $B$ of equal size, whether $B$ dominates $A$.
As there are $\binom{k}{j}$ size-$j$ subsets of a size-$k$ committee and $\binom{m-k}{j}$ size-$j$ subsets of the remaining candidates, checking whether a subset is es-dominated is not feasible to do for every subset, if both the committee size $k$ and number of remaining candidates $m-k$ are large.
The SDO property alleviated this issue, as by the characterization of Pareto optimality if SDO holds only single candidates, i.e., size one candidate subsets, have to be checked for whether they are dominated.
We tried to generalize this by first only considering voting instances for which we only have to check for subsets of at most size two have whether they dominate each other, i.e., voting instances $(C,V,\mathcal A, k)$, for which either $k = 1$ or $m-k \leq 3$ holds.

First, we show that for size-$1$ Pareto optimal committees, committee monotonicity holds.

\begin{proposition}
    Let $(C,V,\mathcal{A}, 1)$ be a voting instance and $W = \{c\} \subseteq C$ be a Pareto optimal committee.

    There exists a candidate $c' \in C \setminus W$ such that $W+c'$ is a size-$2$ Pareto optimal committee.
\end{proposition}
\begin{proof}
    Let $(C,V,\mathcal{A}, 1)$ be a voting instance and $W = \{c\} \subseteq C$ be a Pareto optimal committee.

    Let $c' \in C \setminus W$ be a candidate, such that no other candidate $d \in C \setminus W$ is approved by more voters approving candidate $c$, i.e., $\forall d \in C\setminus W \colon \lvert V(d) \cap V(c) \rvert \leq \lvert V(c') \cap V(c) \rvert$ holds, such that $c'$ is not dominated by any other candidate.

    For the sake of contradiction, assume that the committee $W+c'$ is not Pareto optimal.
    As the committee $W = \{c\}$ is Pareto optimal candidate $c$ is not Pareto dominated by any other candidate.
    Thus, a voter $v^* \in V(c)$ exists who does not approve candidate $c'$.
    Since candidate $c'$ is not dominated by any candidate in $C \setminus W$ there is a set of candidates $\{d_1, d_2\} \subseteq C \setminus (W+c')$ es-dominating $\{c,c'\}$.
    Every voter $v \in V(c) \cap V(c')$ approving both candidate $c$ and $c'$ hence approves candidates $d$ and $d'$ as well.
    Further, voter $v^*$ approves candidate $d$ or $d'$ as $v^*$ approves candidate $c$.
    W.l.o.g. assume $v^*$ approves candidate $d$.
    Then $V(c) \cap V(c') \subset V(c) \cap V(d)$ holds and thus $\lvert V(c) \cap V(d) \rvert > \lvert V(c) \cap V(c') \rvert$ contradicting our choice of candidate $c'$.

    We can conclude that the committee $W+c$ is a size-$2$ Pareto optimal committee.
\end{proof}

In this proof, we explicitly showed that a candidate, for which no other candidate not in the committee is approved by more voters also approving the candidate in the committee, can be added to the committee, resulting in a bigger Pareto optimal committee.
This argument does not extend to larger committee sizes as demonstrated in the counter example in \Cref{fig:general-monotonicity-induction-overlap-counterexample}.

\begin{figure}[h]
    \centering
    \begin{tikzpicture}[yscale=0.5, xscale=0.65, voter/.style={anchor=south}]

    \foreach \i in {1,...,9}
        \node[voter] at (\i-0.5, -1.5) {$\i$};

    \drawcand{1}{1/2}
    \drawcand{2}{3/7}
    \drawcand[violet!30]{3}{1/1,5/6,8/9}
    \drawcand[blue!20]{4}{2/4}
    \drawcand[green!20]{5}{3/6,8/8}

    \end{tikzpicture}
    \caption{For this approval profile $W = \{c_1, c_2\}$ is a Pareto optimal committee where no candidate not in $W$ is approved by more voters approving candidates in the committee than candidate $c_5$. However $\{c_1, c_2, c_5\}$ is not a Pareto optimal size-$3$ committee because the subset $\{c_1, c_5\}$ is es-dominated by $\{c_3, c_4\}$.}
    \label{fig:general-monotonicity-induction-overlap-counterexample}
\end{figure}
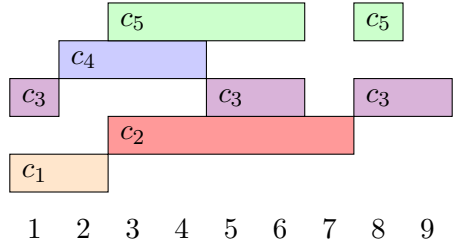

Further, we analyzed committee monotonicity in voting instances for which only three or less candidates are not included in the committee as another way of reducing the number of subsets to consider when checking for Pareto optimality.

\begin{proposition}
    Let $(C, V, \mathcal{A}, k)$ be a voting instance with $1 \leq \lvert C \rvert - k \leq 3$ and $W \subseteq C$ be a size-$k$ Pareto optimal committee.

    There exists a candidate $c \in C \setminus W$ such that $W+c$ is Pareto optimal.
\end{proposition}
\begin{proof}
    Let $(C, V, \mathcal{A}, k)$ be a voting instance with $1 \leq  \lvert C \rvert - k \leq 3$ and $W \subseteq C$ be a size-$k$ Pareto optimal committee.

    We first make a case distinction by the number of candidates not contained in the committee.

    \paragraph{Case 1}
    If only one candidate $c$ is not in the committee $W$, then $W+c$ is a Pareto optimal size-$(k+1)$ committee as it is the only size-$(k+1)$ committee.

    \paragraph{Case 2}
    Assume exactly two candidates $c_1, c_2 \in C$ are not in the committee $W$.
    Either candidate $c_1$ dominates candidate $c_2$, candidate $c_2$ dominates $c_1$ or neither candidate dominates the other.
    Further, neither candidate dominates any candidate in the committee $W$ as the committee is Pareto optimal.

    We distinguish three further cases.

    First, assume that neither candidate $c_1$ nor $c_2$ dominates the other candidate.
    The committee $W+c_1$ is a Pareto optimal size-$(k+1)$ committee as no candidate in $W+c_1$ is dominated by candidate $c_2$, i.e., no subset of the committee $W+c$ is dominated by a subset of the remaining candidates as those are only $\{c_2\}$ and the empty set.
    Thus, the committee $W$ can be extended by another candidate to form a new size-$(k+1)$ Pareto optimal committee.

    Assume that candidate $c_1$ dominates $c_2$.
    Thus, candidate $c_2$ does not dominate candidate $c_1$ and as the committee $W$ is Pareto optimal no candidate in $W$.
    By the same argument as above, $W+c_1$ therefore is a size-$(k+1)$ Pareto optimal committee.

    The remaining case of candidate $c_2$ dominating candidate $c_1$ follows analogously to the previous case.

    Hence, if exactly two candidates are not included in the committee $W$, the committee can be extended to a size-$(k+1)$ Pareto optimal committee.

    \paragraph{Case 3}
    Lastly, assume exactly three candidates $c_1, c_2, c_3 \in C$ are not in the committee $W$.

    W.l.o.g. assume that candidate $c_1$ is not dominated by candidates $c_2$ or $c_3$.
    Further, assume $W+c_1$ is not a Pareto optimal committee.
    Otherwise, $W+c_1$ is a size-$(k+1)$ Pareto optimal committee as sought.
    We can conclude that a candidate $d \in W$ exists such that $\{c_1, d\}$ is dominated by candidates $\{c_2, c_3\}$.

    Let $B,C,D \subseteq V(d)$ be disjoint subsets of the voters such that all voters in $B$ approve candidates $c_1$ and $d$, all voters in $D$ approve candidate $d$ but neither candidate $c_1$ nor $c_2$ and $C$ be the remaining voters approving candidate $d$, as illustrated in \Cref{fig:general-monotonicity-induction-smallrest-works}.
    As the set of candidates $\{c_1,d\}$ is dominated by $\{c_2, c_3\}$, every voter approving candidate $d$ also approves candidate $c_2$ or candidate $c_3$.
    Further, every voter $v \in B$ approves both candidate $c_1$ and candidate $d$ by definition and therefore also has to approve candidates $c_2$ and $c_3$ for $\{c_1, d\}$ to be dominated by $\{c_2, c_3\}$.
    Thus, the voters $v \in C$ approve candidates $d$ and $c_2$ and the voters $v \in D$ approve candidates $d$ and $c_3$.
    As candidate $c_1$ is not dominated by $c_2$ nor $c_3$ and every candidate $v \in B$ approves candidates $c_1, c_2$ and $c_3$ there is a voter $v^* \notin B$ approving candidate $c_1$.
    Voter $v^*$ does therefore not approve candidate $d$ but approves either candidate $c_2$ or $c_3$.
    W.l.o.g. assume that candidate $c_3$ is approved by voter $v^*$.

    We aim to show that $W+c_2$ is a Pareto optimal committee.
    For the sake of contradiction, assume that it is not.
    Thus, every voter $v \in V(c_2)$ approving candidate $c_2$ also approves candidate $c_1$ or $c_3$.
    In particular, every voter $v \in C$ approves candidate $c_3$ as voter $v$ does not approve candidate $c_1$ by the definition of $C$ but must approve candidate $c_1$ or $c_3$.
    As we defined the subsets of voters $B, C, D$ as a partition of all voters approving candidate $d$, $B \cup C \cup D = V(d)$ holds.
    Since every voter $v \in B$ and every voter $v' \in D$ approves candidate $c_3$, we can conclude, that every voter $v \in V(d)$ also approves candidate $c_3$ and hence dominates candidate $d$ contradicting the Pareto optimality of the committee $W$.
    For a visualization of this contradiction see \Cref{fig:general-monotonicity-induction-smallrest-contradiction}.

    Thus, the committee $W+c_2$ is a size-$(k+1)$ Pareto optimal committee.

    \begin{figure}[h]
        \begin{minipage}[c]{0.4\linewidth}
        \centering
            \begin{tikzpicture}[yscale=0.5, xscale=0.65, voter/.style={anchor=south}]

                \foreach \i / \c in {1/A,2/B,3/C,4/D,5/E}
                \node[voter] at (\i-0.5, -1.5) {$\c$};

                \drawcand{1}{1/2,5/5}
                \drawcand{2}{1/3}
                \drawcand[blue!20]{3}{2/2,4/5}

                \draw[fill=green!20] (1,3) rectangle (4,4);
                \node at (1.5,3.5) {$d$};
            \end{tikzpicture}
            \caption{In this figure, $A, \dots, E$ represent sets of voters, remaining candidates of the committee $W$ are omitted. The set $\{c_1, d\}$ is dominated by $\{c_2, c_3\}$ and the voters in $C$ only approve candidate $c_2$ among $\{c_1, c_2, c_3\}$.}
            \label{fig:general-monotonicity-induction-smallrest-works}
        \end{minipage}
        \hfill
        \begin{minipage}[c]{0.4\linewidth}
        \centering
            \begin{tikzpicture}[yscale=0.5, xscale=0.65, voter/.style={anchor=south}]

                \foreach \i / \c in {1/A,2/B,3/C,4/D,5/E}
                \node[voter] at (\i-0.5, -1.5) {$\c$};

                \drawcand{1}{1/2,5/5}
                \drawcand{2}{1/3}
                \drawcand[blue!20]{3}{2/5}

                \draw[fill=green!20] (1,3) rectangle (4,4);
                \node at (1.5,3.5) {$d$};
            \end{tikzpicture}
            \caption{In this figure, $A, \dots, E$ represent sets of voters, remaining candidates of the committee $W$ are omitted. The set $\{c_1, d\}$ is dominated by $\{c_2, c_3\}$ and candidate $c_3$ dominates $d$ contradicting the Pareto optimality of $W$.}
            \label{fig:general-monotonicity-induction-smallrest-contradiction}
        \end{minipage}
    \end{figure}

    In every case, the committee $W$ can be extended to a size-$(k+1)$ Pareto optimal committee.
    Hence, our claim holds.
\end{proof}

Moving on to larger numbers of either candidates in the committee or candidates not in the committee becomes increasingly more challenging.
In the proof above, for case three, we partitioned voters into voters approving only one candidate of the committee and voters approving two candidates of the committee in the last case.
We made use of the fact that approving two candidates means approving everyone in the committee.
When dealing with larger subsets of the committee and checking whether they are es-dominated by a subset of the remaining candidates, the number of ways a voter can approve e.g. two candidates in such a subset increases and thus less can be inferred about the candidates in that subset.

Further, it becomes impractical to use sets to represent voters approving candidates, as it is not only relevant whether a voter approves a candidate in a set of candidates but also how many for Pareto optimality.
Instead multisets or other ways of representing which voter approves how many candidates in a set of candidates have to be used.
E.g., which voter approves a candidate can be represented by a binary value vector with one entry per voter.

The question whether committee monotonicity holds for Pareto optimal committees in general thus remains open.

\subsection{Reconfiguration Graph}
\label{chp:general-reconfiguration}

Reconfiguring Pareto optimal committees also becomes substantially more complex in the unrestricted domain.
As explained in \Cref{chp:general-monotonicity}, verifying Pareto optimality of committees in general requires considering all subsets of the committee, not only one newly added candidate which our proofs in \Cref{subsec:sdo-impl-reconfig-graph} relied on.
To reconfigure a Pareto optimal size-$k$ committee $W_1$ into another one $W_2$, for voting instances satisfying SDO, we first removed a candidate from $W_1$ who is not in $W_2$ and has a minimum approval score.
This size-$(k-1)$ intermediate committee is Pareto optimal as the candidate removed does not dominate any candidate in $W_1$ and the SDO property holds.
By then adding a candidate from $W_2$ who is not in $W_1$ and has a maximum approval score, and, hence, is not Pareto dominated by any candidate not in $W_1$, we merge two Pareto optimal committees.
This resulting committee is Pareto optimal as well because the SDO property holds.
In general, however, the union of two Pareto optimal committees is not necessarily Pareto optimal, as explained in \Cref{chp:general-monotonicity}.

Thus, our proof cannot be extended to the unrestricted domain.

\subsubsection{Proving Connectedness}

An important open question we were not able to answer is whether the Pareto optimality reconfiguration graph is connected in the natural domain as well.
In reconfiguration problems, a common technique of proving connectedness is determining a \emph{canonical configuration} and proving that every object, in this case Pareto optimal committees, can be reconfigured into the canonical configuration \cite{introtoreconf}.
We propose to use this method to show the connectedness of the Pareto optimality reconfiguration graph $\Gamma_{Po}(\mathcal{A}, k)$ by choosing any size-$k$ committee of maximum approval score as the canonical configuration.
As the committee has a maximum approval score among all size-$k$ committees it must be Pareto optimal and thus in $\Gamma_{Po}(\mathcal{A}, k)$.
In \Cref{lem:unrestricted-reconf-dist-maxapprovalscore}, we prove that any size-$k$ committee of maximum approval score can be reconfigured into any other size-$k$ committee of maximum approval score in at most $k$ steps.
Thus, the canonical configuration chosen is connected to every Pareto optimal of maximum approval score, by which we can conclude that all Pareto optimal committees of maximum approval score are connected.
To show the connectedness of the entire reconfiguration graph, it remains to be shown that any size-$k$ Pareto optimal committee can be transformed into any size-$k$ committee of maximum approval score, which we expect to be true.
Showing this would yield that the Pareto optimality reconfiguration graph is connected for every voting instance, as every Pareto optimal committee could thus be reconfigured into the canonical configuration.

\begin{lemma}
    Let $W_1, W_2 \subseteq C$ be two size-$k$ committees of maximum approval score for a voting instance $(C, V, \mathcal{A})$.
    The distance of $W_1$ and $W_2$ in the reconfiguration graph $\Gamma_{Po}(\mathcal{A}, k)$ is at most $k$.
    \label{lem:unrestricted-reconf-dist-maxapprovalscore}
\end{lemma}
\begin{proof}
    Let $(C, V, \mathcal{A})$ be a voting instance and $W_1, W_2 \subseteq C$ two size-$k$ committees.

    Let $c_1,\in W_1 \setminus W_2$ and $c_2 \in W_2 \setminus W_1$ be two candidates not contained in both committees.
    Both candidates have the same approval score.
    For the sake of contradiction and w.l.o.g. assume that candidate $c_1$ has a higher approval score than $c_2$.
    Hence, the committee $W_2 - c_2 + c_1$ has a higher approval score than the committee $W_2$ contradicting the choice of $W_2$.
    Since both candidates have the same approval score and $W_1$ is a committee of maximum approval score, the committee $W_1 - c_1 + c_2$ is also of maximum approval score and thus Pareto optimal.
    By repeatedly replacing each candidate in $W_1$ but not in $W_2$ by a candidate in $W_2$ but not in $W_1$, the committee $W_1$ therefore remains Pareto optimal in all reconfiguration steps, $\lvert W_1 \setminus W_2 \rvert$ many of which thus have to be carried out.
    Hence, the distance of the committees $W_1$ and $W_2$ in the Pareto optimality reconfiguration graph is at most $k$.
\end{proof}

As for reconfiguring a non-maximum approval score committee $W_1$ into a maximum approval score one $W_2$, we expect it to be possible to replace a candidate of $W_1$ by a candidate of $W_2$ with a higher approval score.

\begin{conjecture}
    Let $(C, V, \mathcal{A}, k)$ be a voting instance and $W \subseteq C$ a Pareto optimal committee such that another committee $W' \subseteq C$ with a higher approval score exists.

    There exist candidates $c \in W$ and $d \in C \setminus W$ with $\lvert V(c)\rvert < \lvert V(d) \rvert$, such that $W - c + d$ is a Pareto optimal committee.
\end{conjecture}

While it is fairly easy to prove this for a committee $W$ in which only one candidate $c \in W$ exists such that a candidate $d \in C \setminus W$ with $\lvert V(c) \rvert < \lvert V(d) \rvert$ exists, we were not able to show this to be true for all committees.
An intuitive approach would be using the candidate with the highest approval score not in the committee as a replacement for a candidate with a lower approval score in the committee.
However, this does not work as demonstrated in the following.

\begin{figure}[h]
\centering
    \begin{tikzpicture}[yscale=0.4,xscale=0.65, voter/.style={anchor=south}]

        \foreach \i in {1,...,18}
    		\node[voter] at (\i-0.5, -1.5) {$\i$};

    \draw[fill=orange!20] (0,0) rectangle (4,1);
    \node at (0.5,0.5) {$c_1$};

    \draw[fill=red!40] (4,1) rectangle (8,2);
    \node at (4.5,1.5) {$c_2$};

    \draw[fill=magenta!20] (8,2) rectangle (16,3);
    \node at (8.5,2.5) {$c_3$};

    \draw[fill=violet!30] (0,3) rectangle (2,4);
    \draw[fill=violet!30] (8,3) rectangle (12,4);
    \node at (0.5,3.5) {$c_4$};
    \node at (8.5,3.5) {$c_4$};

    \draw[fill=blue!20] (2,4) rectangle (4,5);
    \draw[fill=blue!20] (12,4) rectangle (17,5);
    \node at (2.5,4.5) {$c_5$};
    \node at (12.5,4.5) {$c_5$};

    \draw[fill=green!20] (4,5) rectangle (6,6);
    \draw[fill=green!20] (8,5) rectangle (12,6);
    \node at (4.5,5.5) {$c_6$};
    \node at (8.5,5.5) {$c_6$};

    \draw[fill=yellow!30] (6,6) rectangle (8,7);
    \draw[fill=yellow!30] (12,6) rectangle (18,7);
    \node at (6.5,6.5) {$c_7$};
    \node at (12.5,6.5) {$c_7$};

    \end{tikzpicture}
    \caption{The committee $W = \{c_1,c_2\}$ is Pareto optimal in this voting instance. Candidate $c_3$ is the candidate with the highest approval score not in the committee $W$. However, both $\{c_1, c_3\}$ and $\{c_2, c_3\}$ are not Pareto optimal. Thus, replacing a candidate of the committee $W$ by a candidate with maximum approval score is not possible while preserving Pareto optimality.}
    \label{fig:general-reconf-maxscore-counterexample}
\end{figure}
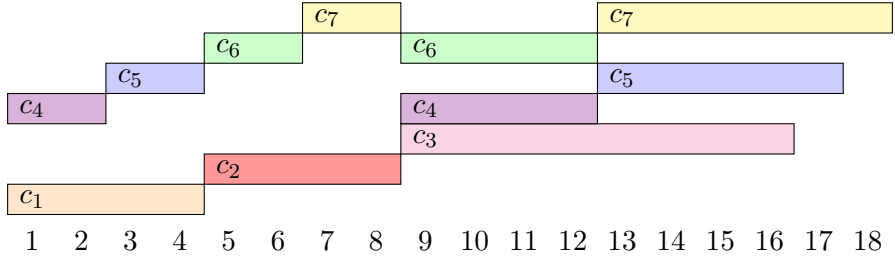

Consider the voting instance depicted in \Cref{fig:general-reconf-maxscore-counterexample} and the committee $W = \{c_1,c_2\}$.
As no other pair of candidates dominates $W$, the committee $W$ is Pareto optimal.
Further, candidate $c_3$ is the candidate with the highest approval score and not in the committee $W$.
Hence, we would want to replace candidates $c_1$ or $c_2$ in $W$ by candidate $c_3$.
However, the committee $\{c_1, c_3\}$ is dominated by the committee $\{c_4, c_5\}$ and $\{c_2, c_3\}$ is dominated by $\{c_6, c_7\}$.
Thus, candidate $c_3$ cannot be used as a replacement for a candidate of $W$ when trying to maintain Pareto optimality.

We want to give an intuition for why this does not work:
Considering why the committee $\{c_1, c_3\}$ is not Pareto optimal, we can see that the candidates $c_4$ and $c_5$ also have a rather high approval score and together are approved by the same voters approving candidate $c_1$.
A set of candidates that dominates another set of candidates must always have a higher approval score than the dominated set.
However, as the candidate sets $\{c_1, c_3\}$ and $\{c_4, c_5\}$ both consist of two candidates, not every candidate in the dominating set must have a higher approval score than each candidate in the dominated set.
Because the voters approving candidate $c_4$ and $c_5$ are distributed equally across the two candidates, which is not the case for candidates $c_1$ and $c_3$, the candidate set $\{c_4, c_5\}$ can still dominate $\{c_1, c_3\}$.

Hence, another criterion for picking a candidate to replace a candidate in the original committee with, such that the approval score of the committee increases and Pareto optimality is preserved, has to be found.

We believe this open problem to be closely related to the committee monotonicity problem described in \Cref{chp:general-monotonicity}, as this problem can also be viewed as removing one candidate from the size-$k$ committee such that the remaining committee is a Pareto optimal size-$(k-1)$ committee and adding back a different candidate to the committee, forming a new size-$k$ Pareto optimal committee.
This requires first studying, for each Pareto optimal size-$k$ committee $W$, which candidates $c \in W$ can be removed from the committee, such that $W-c$ is a Pareto optimal size-$(k-1)$ committee.
Further, not only whether committee monotonicity holds, but also which candidates can be added to Pareto optimal committees while preserving Pareto optimality has to be studied.

Vice versa, proving the connectedness of the reconfiguration graph may make it easier to prove that committee monotonicity holds for Pareto optimal committees:
If any Pareto optimal size-$k$ committee $W$ can be reconfigured into any other, then there must be a first reconfiguration step in which a candidate $c \in W$ is replaced by a candidate $d \in C \setminus W$, such that the committee $W-c+d$ is Pareto optimal.
Thus, no subset $A \subseteq W+d$, for which $c \notin A$ or $d \notin A$ holds, is es-dominated by a subset of the remaining candidates, which only leaves subsets $B \subseteq W+d$, for which $c \in B$ and $d \in B$ holds to consider, when checking whether the committee $W+d$ is Pareto optimal.

This, however, does not imply that either both committee monotonicity holds and the reconfiguration graph is connected or neither holds - one may be true while the other is false, despite the similarity of the two questions.

\subsubsection{Distance of Committees}

Regardless of whether the Pareto optimality reconfiguration graph is connected, we want to demonstrate that the distance of two size-$k$ Pareto optimal committees $W_1, W_2 \subseteq C$ in the reconfiguration graph can be more than $k - \lvert W_1 \cap W_2 \rvert$, proven to be the distance of two committees in the CI and VI domains in \Cref{cor:restricted-candidate-interval-reconfiguration} and \Cref{cor:restricted-voter-interval-reconfiguration}, respectively.

\begin{proposition}
    In approval-based multiwinner voting the distance between two Pareto optimal committees $W_1$ and $W_2$ in the reconfiguration graph $\Gamma_{Po}(\mathcal{A}, k)$ can be strictly greater than $k - |W_1 \cap W_2|$.
\end{proposition}
\begin{proof}
    To prove the proposition, we give a counter example:

    Let $(C, V, \mathcal{A}, 2)$ be a voting instance with an approval profile as defined in \Cref{fig:general-reconf-dist-counterexample}

\begin{figure}[h]
\centering
\begin{tikzpicture}[yscale=0.5, xscale=0.65, voter/.style={anchor=south}]

    \foreach \i in {1,...,11}
        \node[voter] at (\i-0.5, -1.5) {$\i$};

    \draw[fill=orange!20] (0,0) rectangle (4,1);
    \node at (0.5,0.5) {$c_1$};

    \draw[fill=red!40] (0,1) rectangle (2,2);
    \draw[fill=red!40] (4,1) rectangle (6,2);
    \node at (0.5,1.5) {$c_2$};
    \node at (4.5,1.5) {$c_2$};

    \draw[fill=purple!35] (2,2) rectangle (3,3);
    \draw[fill=purple!35] (6,2) rectangle (8,3);
    \node at (2.5,2.5) {$c_3$};
    \node at (6.5,2.5) {$c_3$};

    \draw[fill=magenta!20] (4,3) rectangle (5,4);
    \draw[fill=magenta!20] (8,3) rectangle (10,4);
    \node at (4.5,3.5) {$c_4$};
    \node at (8.5,3.5) {$c_4$};

    \draw[fill=violet!30] (0,4) rectangle (3,5);
    \draw[fill=violet!30] (4,4) rectangle (5,5);
    \node at (0.5,4.5) {$c_5$};
    \node at (4.5,4.5) {$c_5$};

    \draw[fill=blue!20] (5,5) rectangle (8,6);
    \draw[fill=blue!20] (10,5) rectangle (11,6);
    \node at (5.5,5.5) {$c_6$};
    \node at (10.5,5.5) {$c_6$};

    \draw[fill=cyan!20] (0,6) rectangle (1,7);
    \draw[fill=cyan!20] (2,6) rectangle (4,7);
    \draw[fill=cyan!20] (6,6) rectangle (7,7);
    \node at (0.5,6.5) {$c_7$};
    \node at (2.5,6.5) {$c_7$};
    \node at (6.5,6.5) {$c_7$};

    \draw[fill=teal!20] (1,7) rectangle (3,8);
    \draw[fill=teal!20] (7,7) rectangle (8,8);
    \draw[fill=teal!20] (10,7) rectangle (11,8);
    \node at (1.5,7.5) {$c_8$};
    \node at (7.5,7.5) {$c_8$};
    \node at (10.5,7.5) {$c_8$};

    \draw[fill=green!20] (0,8) rectangle (1,9);
    \draw[fill=green!20] (4,8) rectangle (6,9);
    \draw[fill=green!20] (8,8) rectangle (9,9);
    \node at (0.5,8.5) {$c_9$};
    \node at (4.5,8.5) {$c_9$};
    \node at (8.5,8.5) {$c_9$};

    \draw[fill=lime!20] (1,9) rectangle (2,10);
    \draw[fill=lime!20] (4,9) rectangle (5,10);
    \draw[fill=lime!20] (9,9) rectangle (11,10);
    \node at (1.5,9.5) {$c_{10}$};
    \node at (4.5,9.5) {$c_{10}$};
    \node at (9.5,9.5) {$c_{10}$};

    \draw[fill=yellow!30] (3,10) rectangle (4,11);
    \draw[fill=yellow!30] (8,10) rectangle (11,11);
    \node at (3.5,10.5) {$c_{11}$};
    \node at (8.5,10.5) {$c_{11}$};

\end{tikzpicture}
\caption{For this approval profile, $W_1=\{c_1,c_2\}$ and $W_2=\{c_3,c_4\}$ are Pareto optimal committees with a distance of 3 in the reconfiguration graph.}
\label{fig:general-reconf-dist-counterexample}
\end{figure}
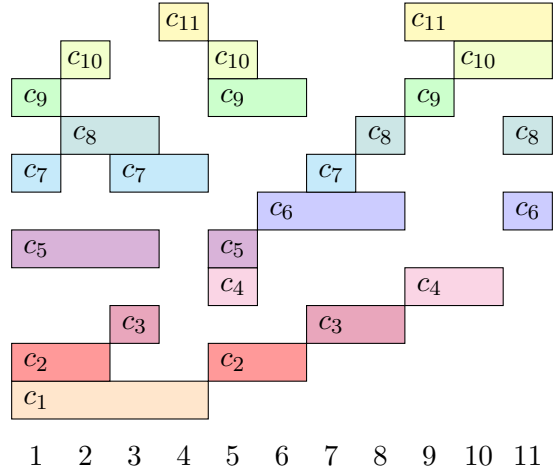

    The committees $W_1 = \{c_1, c_2\}$ and $W_2 = \{c_3, c_4\}$ both are Pareto optimal committees.
    However, $W_2$ can only be reached from $W_1$ doing at least 3 swaps of candidates.

    All cross-combinations of two candidates from $W_1$ and $W_2$ are es-dominated by some set of candidates and thus not Pareto optimal:

    The committee $\{c_2, c_3\}$ is dominated by the set $\{c_5, c_6\}$, the committee $\{c_1, c_3\}$ is dominated by the set $\{c_7, c_8\}$, the committee $\{c_2, c_4\}$ is dominated by the set $\{c_9, c_{10}\}$, and the committee $\{c_1, c_4\}$ is dominated by the set $\{c_5, c_{11}\}$.

    Thus, the every way of reconfiguring the committee $W_1$ into the committee $W_2$ must involve replacing either candidate $c_1$ or $c_2$ by a candidate not in $W_2$, e.g. candidate $c_{11}$, then replacing the other candidate originally in $W_1$ by a candidate of $W_2$ and lastly replacing the previously added candidate, like candidate $c_{11}$.
    The distance of the committees $W_1$ and $W_2$ therefore is at least $3$.
\end{proof}

This implies that Pareto optimal committees can not always be reconfigured into any other Pareto optimal committee on the most direct way.
Instead, for some committees, a candidate that is neither in the starting committee nor the committee to be reached has to be included into the committee in one reconfiguration step and then later be replaced in another step.

If such committees exist, this implies that committees with larger than average neighborhoods exist in the graph.
For these committees, multiple different Pareto optimal committees exist, which only differ from that committee in one candidate, implying that there is a subset of the committee that is difficult to dominate.
Finding these committees might thus prove useful for showing committee monotonicity as well.

\section{Conclusion \& Outlook}

We investigated the structural properties of Pareto optimal committees.
We focused in particular on the Candidate Interval and Voter Interval, where we characterized Pareto optimality using the proposed Single Dominance Only property.
Our results show that Pareto optimal committees are both reconfigurable and extendable, and we suspect that in the Voter Interval domain they can additionally be counted in polynomial time.
Together, this establishes a strong inherent structure in these restricted domains.
Further, we adapted a polynomial-time algorithm for finding committees satisfying EJR+ \citep{PropAxiomsForMultiwVoting} to additionally guarantee Pareto optimality for restricted domains, thus not only resolving the open question of finding such an algorithm for the Voter Interval domain, but also showing how Pareto optimal committees satisfying further properties can be constructed.

For the unrestricted domain, we highlighted the challenges of generalizing these results.
Despite presenting a counterexample showing that extending a Pareto optimal committee by a candidate based on approval scores does not always preserve Pareto optimality, we suspect committee monotonicity to hold more generally.
Our proofs for small committees and voting instances with few candidates not in the committee suggest a structure of Pareto optimal committees in voting instances beyond the Candidate Interval and Voter Interval domains.

Moreover, we demonstrated that the distance of two committees $A$ and $B$ in the Pareto optimality reconfiguration graph can exceed $k-\lvert A \cap B \rvert$, and proposed an approach towards proving connectedness.
The question of connectedness in the general case, however, remains unresolved.

Building on our results for the restricted domains, future work could make use of the structure of Pareto optimal committees to find algorithms for committees satisfying both Pareto optimality and further axioms simultaneously.
Besides proving the correctness of our proposed algorithm for counting the number of Pareto optimal committees for voting instances satisfying Voter Interval, another interesting direction to follow is counting the number of such committees for voting instances satisfying Candidate Interval, which we suspect is possible in polynomial time as well.

Additionally, future research could explore relaxations of the Single Dominance Only property or alternative characterizations of Pareto optimality to study committee monotonicity and reconfigurability in other restricted domains.
The connectedness of winning committees according to different popular voting rules specifically, similarly to the work of \cite{reconfiguration-jr}, or finding characterizations of committees known to be connected could also pose interesting questions and may lay ground for developing algorithms for finding Pareto optimal committees satisfying other axioms as well.

\newpage

\paragraph{Acknowledgement}
I would like to thank my supervisors, Niclas Böhmer and Lara Glessen, for their guidance, insights, and ongoing feedback throughout my work on this thesis. Discussing different ideas and proofs with them opened my eyes to many approaches and details I might otherwise have missed. Their suggestions on which directions to pursue, along with their introduction of new ideas, consistently helped me make progress.

\end{document}